%% file: ctodadispersionless-arxiv-rev.tex
\documentclass[12pt]{article}
\usepackage{amsmath,amssymb,amsthm,amscd}
\usepackage[dvipdfmx]{graphicx,color}
\usepackage{import}
\usepackage{mathrsfs}
\def\hybrid{\topmargin 0pt      \oddsidemargin 0pt
        \headheight 0pt \headsep 0pt
       \voffset-1cm
        \textwidth 6.25in       % A4 paper
       \textheight 9.5in       % A4 paper
        \marginparwidth 0.0in
        \parskip 5pt plus 1pt   \jot = 1.5ex}
\catcode`\@=11
\def\marginnote#1{}

\newcount\hour
\newcount\minute
\newtoks\amorpm
\hour=\time\divide\hour by60
\minute=\time{\multiply\hour by60 \global\advance\minute by-\hour}
\edef\standardtime{{\ifnum\hour<12 \global\amorpm={am}%
        \else\global\amorpm={pm}\advance\hour by-12 \fi
        \ifnum\hour=0 \hour=12 \fi
        \number\hour:\ifnum\minute<10 0\fi\number\minute\the\amorpm}}
\edef\militarytime{\number\hour:\ifnum\minute<10 0\fi\number\minute}

\def\draftlabel#1{{\@bsphack\if@filesw {\let\thepage\relax
   \xdef\@gtempa{\write\@auxout{\string
      \newlabel{#1}{{\@currentlabel}{\thepage}}}}}\@gtempa
   \if@nobreak \ifvmode\nobreak\fi\fi\fi\@esphack}
        \gdef\@eqnlabel{#1}}
\def\@eqnlabel{}
\def\@vacuum{}
\def\draftmarginnote#1{\marginpar{\raggedright\scriptsize\tt#1}}

\def\draftlabel#1{{\@bsphack\if@filesw {\let\thepage\relax
   \xdef\@gtempa{\write\@auxout{\string
      \newlabel{#1}{{\@currentlabel}{\thepage}}}}}\@gtempa
   \if@nobreak \ifvmode\nobreak\fi\fi\fi\@esphack}
        \gdef\@eqnlabel{#1}}
\def\@eqnlabel{}
\def\@vacuum{}
\def\draftmarginnote#1{\marginpar{\raggedright\scriptsize\tt#1}}

\def\draft{\oddsidemargin -.5truein
        \def\@oddfoot{\sl preliminary draft \hfil
        \rm\thepage\hfil\sl\today\quad\militarytime}
        \let\@evenfoot\@oddfoot \overfullrule 3pt
        \let\label=\draftlabel
        \let\marginnote=\draftmarginnote
   \def\@eqnnum{(\theequation)\rlap{\kern\marginparsep\tt\@eqnlabel}%
\global\let\@eqnlabel\@vacuum}  }

\def\numberbysection{\@addtoreset{equation}{section}
        \def\theequation{\thesection.\arabic{equation}}}

\def\underline#1{\relax\ifmmode\@@underline#1\else
        $\@@underline{\hbox{#1}}$\relax\fi}

\def\titlepage{\@restonecolfalse\if@twocolumn\@restonecoltrue\onecolumn
     \else \newpage \fi \thispagestyle{empty}\c@page\z@
        \def\thefootnote{\fnsymbol{footnote}} }

\def\endtitlepage{\if@restonecol\twocolumn \else  \fi
        \def\thefootnote{\arabic{footnote}}
        \setcounter{footnote}{0}}  %\c@footnote\z@ }
\relax

\numberbysection
\hybrid

\newfont{\Bbbb}{msbm7 scaled 1\@ptsize00}
\newcommand{\CC}{\mathbb C}

\newcommand{\ccc}{\raise-1pt\hbox{$\mbox{\Bbbb C}$}}
\newcommand{\DD}{\mathbb D}
\newcommand{\DDD}{\raise-1pt\hbox{$\mbox{\Bbbb D}$}}

\newcommand{\RR}{\mathbb R}

\newcommand{\UU}{\mathbb U}
\newcommand{\UUU}{\raise-1pt\hbox{$\mbox{\Bbbb U}$}}

\newcommand{\ZZ}{\mathbb Z}
\newcommand{\z}{\raise-1pt\hbox{$\mbox{\Bbbb Z}$}}

\newcommand{\SSS}{\mathbb S}
\newcommand{\sss}{\raise-1pt\hbox{$\mbox{\Bbbb S}$}}

\def\beq{\begin{equation}}
\def\eeq{\end{equation}}
\def\p{\partial}

\newcommand\bbD{\mathbb{U}}
\newcommand\Comp{\mathbb{C}}
\newcommand\der{\partial}

\newcommand\interior[1]{\overset\circ{#1}}

\renewcommand\Im{\mathrm{Im}\,}

\renewcommand\Re{\mathrm{Re}\,}
\newcommand\Real{\mathbb{R}}
\renewcommand\setminus{\smallsetminus}

\newcommand\secref[1]{Section~\ref{#1}}
\newcommand\appref[1]{Appendix~\ref{#1}}

\newcommand\lemref[1]{Lemma~\ref{#1}}

\newcommand\figref[1]{Figure~\ref{#1}}
\newtheorem{theorem}{Theorem}[section]
\newtheorem{lemma}{Lemma}[section]
\newtheorem{lemma-definition}{Lemma-Definition}[section]

\begin{document}

\begin{titlepage}

\title{Dispersionless version of the constrained Toda hierarchy
and symmetric radial L\"owner equation}

\author{T. Takebe\thanks{
National Research University Higher School of
Economics,
20 Myasnitskaya Ulitsa, Moscow 101000, Russian Federation,
e-mail: ttakebe@hse.ru}
\and
A.~Zabrodin\thanks{
Skolkovo Institute of Science and Technology, 143026, 
Moscow, Russian Federation and
Steklov Mathematical Institute of Russian Academy of Sciences,
Gubkina str. 8, Moscow, 119991, Russian Federation,
e-mail: zabrodin@itep.ru}}

\date{April 2022}
\maketitle

%\vspace{-7cm} \centerline{ \hfill ITEP-TH-??/21}\vspace{7cm}

\begin{abstract}

We study the dispersionless version of the recently introduced constrained Toda
hierarchy. Like the Toda lattice itself, it admits three equivalent formulations:
the formulation in terms of Lax equations, the formulation of the 
Zakharov-Shabat type
and the formulation through the generating equation for the 
dispersionless limit of 
logarithm of the tau-function. We show that the dispersionless 
constrained Toda
hierarchy describes conformal maps of reflection-symmetric planar 
domains to the exterior
of the unit disc. We also find finite-dimensional reductions 
of the hierarchy and 
show that they are characterized by a differential equation 
of the L\"owner type which we call the symmetric radial L\"owner equation.
It is also shown that solutions to the 
symmetric radial L\"owner equation are conformal maps of the exterior of
the unit circle with two symmetric slits to the exterior of the unit 
circle. 

\end{abstract}

\end{titlepage}

\vspace{5mm}

%

%\newpage
\tableofcontents

\vspace{5mm}

\section{Introduction}

The constrained Toda hierarchy was recently introduced in \cite{KZ21}. It is 
a subhierarchy of the Toda lattice hierarchy \cite{UT84} defined by the 
constraint 
\beq\label{int1}
\bar L =L^{\dag}
\eeq
for the two pseudo-difference Lax operators $L$, $\bar L$ (in the symmetric gauge). 
The constraint is preserved by the flows $\p_{t_k}-\p_{t_{-k}}$ and is destroyed
by the flows $\p_{t_k}+\p_{t_{-k}}$ of the Toda lattice hierarchy, so one has to
fix $t_k +t_{-k}=0$. Here $\{t_k\}_{k\in \z}$ is the infinite set of independent variables
(times) indexed by integer numbers. The times $t_k$ with $k\neq 0$ are
in general complex numbers while $t_0=n$ is special: it takes integer values. 
The hierarchy is defined by the infinite set of Lax equations (evolution equations 
for the Lax operator in the times $t_k$). 

An equivalent formulation is via the tau-function $\tau$ which is a function of the 
infinite set of independent variables satisfying functional relations which are
bilinear in the case of the Toda lattice and have a more complicated form 
for the constrained Toda hierarchy. 
In order to perform the dispersionless limit \cite{TT91,TT95}, one should substitute
\beq\label{int2}
t_k \longrightarrow \frac{t_k}{\hbar}, \qquad \tau (\{t_i/\hbar \} )=
e^{F(\{t_i\}, \hbar )/\hbar^2}
\eeq
and take the $\hbar \to 0$ limit which allows one to obtain an equation for the function 
$F=\lim\limits_{\hbar \to 0}F(\{t_i\}, \hbar )$ from the 
equation for the tau-function.

Around 2000 it was observed that the dispersionless 
hierarchies are deeply related to
the theory of univalent functions and in particular to conformal maps of planar domains.
This observation was developed in two seemingly different but related 
directions. One of them treats
equations of the dispersionless Toda hierarchy as governing equations for conformal maps
of planar simply connected domains with smooth boundary as functions of their harmonic
moments which are identified with the hierarchical times $t_k$ \cite{MWWZ00,WZ00}. 
Another one is related to conformal maps of domains with curved slits. In the seminal 
papers \cite{GT1,GT2} it was shown that one-variable reductions of the 
dispersionless
Kadomtsev-Petviashvili (KP) hierarchy are classified by solutions of a 
L\"owner-type
differential equation which characterizes one-parameter families of conformal mappings
of domains with a growing slit onto a fixed reference domain 
\cite{L1923}--\cite{review}. Later
this important observation was extended to hierarchies of other types and other versions 
of the L\"owner equation \cite{Manas1}--\cite{ATZ21}.

The aim of this paper is to suggest the dispersionless version of the constrained
Toda hierarchy (Section 3) 
and clarify its relation to conformal maps (Section 4). We will show that 
equations of the dispersionless constrained Toda hierarchy govern conformal maps
of domains symmetric under reflection with respect to the 
real axis (reflection-symmetric domains).

Furthermore, 
we study finite-dimensional reductions of the 
dispersionless constrained Toda hierarchy (Sections 5,6).
The meaning of reduction is as follows. 
Infinite hierarchies of partial differential equations contain an infinite number of
independent variables (``times'') and an infinite number of dependent variables. 
The simplest possible reduction (one-variable reduction) 
is a reduction to just one dependent variable which
depends on all the times. All other dependent
variables become functions of it. 
We will show that the one-variable reductions of the 
dispersionless constrained Toda hierarchy are described
by solutions of a single ordinary differential equation of L\"owner type
which we call symmetric L\"owner equation. 
The equation contains an arbitrary function
(the ``driving function'')
which characterizes the reduction. 
Geometrically, the driving function 
characterizes the shape of the slit and 
the single dependent variable is a parameter along the slit. 
One can also consider multivariable 
($N$-variable) reductions
when there are $N$ dependent variables. In this case the reduction is described
by a system of $N$ symmetric L\"owner equations with $N$ driving functions.

Finally, Section 7 is devoted to clarifying the geometric meaning of the
symmetric radial L\"owner equation. We show that
solutions of this equation are conformal maps of domains 
with two curved slits symmetric under reflection with respect to
the real axis. In the appendix we give a rigorous proof of the 
symmetric radial L\"owner equation. 

\section{Constrained Toda hierarchy}

We begin with the Toda lattice hierarchy \cite{UT84} in the symmetric gauge
\cite{Takebe1,Takebe2}. The two Lax operators are
\beq\label{t1}
L=c(n)e^{\p _n}+\sum_{k\geq 0}U_k(n)e^{-k\p_n}, \quad
\bar L=c(n-1)e^{-\p _n}+\sum_{k\geq 0}\bar U_k(n)e^{k\p_n},
\eeq
where $e^{k\p_n}$ are shift operators acting on functions of $n$ as $e^{k\p_n}f(n)=f(n+k)$ and $c(n)$, $U_k(n)$, $\bar U_k(n)$ are some functions of $n$.
The Lax operators are pseudo-difference operators (infinite Laurent series in the 
shift operator). Given the Lax operators, one can introduce the difference operators
\beq\label{t2}
B_m =(L^m)_{>0} +\frac{1}{2} (L^m)_{0}, \quad
B_{-m}=(\bar L^m)_{<0}+\frac{1}{2} (\bar L^m)_{0}, \quad m=1,2,3, \ldots ,
\eeq
where for a subset $\SSS \subset \ZZ$, we denote
$\displaystyle{\Bigl (\sum_{k\in \z} U_k e^{k \p_n}\Bigr )_{\sss}=
\sum_{k\in \sss} U_k e^{k \p_n}}$. The Toda lattice hierarchy is given by the 
Lax equations
\beq\label{t3}
\p_{t_m}L=[B_m, L], \quad \p_{t_m}\bar L=[B_m, \bar L]
\eeq
which define the hierarchical flows parametrized by the times $t_m$ for any non-zero
integer $m$. These equations seem to be disconnected but actually they
are connected via the common function $c(n)$.
An equivalent formulation is through the zero curvature (Zakharov-Shabat)
equations
\beq\label{t4}
\p_{t_k}B_m-\p_{t_m}B_k +[B_m, B_k]=0.
\eeq
They encode differential-difference equations for the coefficient functions
$c(n)$, $U_k(n)$, $\bar U_k(n)$.

The constrained Toda hierarchy is obtained by imposing the constraint
\beq\label{t5}
\bar L =L^{\dag},
\eeq
where the ${}^\dag$-operation is defined as $(f(n)\circ e^{k\p_n})^{\dag}
=e^{-k\p_n}\circ f(n)$. This implies that $\bar U_k(n)=U_k(n+k)$. It is easy to see that
the constraint is preserved by the flows
$\p_{t_k}-\p_{t_{-k}}$ and is destroyed
by the flows $\p_{t_k}+\p_{t_{-k}}$, so one has to
fix $t_k +t_{-k}=0$. In this way all the coefficient functions can be regarded as functions of
$t_k$ with $k>0$ only. Introducing difference operators
\beq\label{t6}
A_m =B_m-B_{-m},
\eeq
we can write the Lax and Zakharov-Shabat equations of the constrained hierarchy
in the form
\beq\label{t7}
\p_{t_m}L=[A_m, L], \quad [\p_{t_m}-A_m, \, \p_{t_k}-A_k]=0, \quad m>0.
\eeq
The simplest nontrivial equation of the hierarchy (obtained at $m=1$, $k=2$)
is
\beq\label{equation}
\begin{array}{l}
(\p_{t_2}-\p_{t_1}^2)\varphi_{n+1}-(\p_{t_2}+\p_{t_1}^2)\varphi_{n}
\\ \\
\phantom{aaaaaaaaaaaaaaa}
=2e^{\varphi_n -\varphi_{n-1}}-2e^{\varphi_{n+2} -\varphi_{n+1}}+
\frac{1}{2}(\p_{t_1}\varphi_{n+1})^2-\frac{1}{2}(\p_{t_1}\varphi_{n})^2,
\end{array}
\eeq
where the function $\varphi_n$ is defined as $c(n)=e^{\frac{1}{2}
(\varphi_{n+1}-\varphi_n )}$.

The Zakharov-Shabat and Lax equations are compatibility conditions for the 
linear problems
\beq\label{t8}
\p_{t_m}\psi =A_m\psi , \quad L\psi =z\psi 
\eeq
for the wave function $\psi =\psi (n,{\bf t}, z )$ 
depending on the spectral parameter $z\in \CC$
(and on all the times ${\bf t}=\{t_1, t_2, t_3 , \ldots \}$). The set of linear
equations is equivalent to the bilinear relation \cite{KZ21}
\beq\label{t9}
\Bigl ( \oint_{C_{\infty}}-\oint_{C_0}\Bigr )\psi (n, {\bf t}, z)\psi (n', {\bf t}', z^{-1})
\frac{dz}{2\pi i z}=0
\eeq
valid for all $n,n', {\bf t}, {\bf t}'$, 
where $C_{\infty}$, $C_0$ are small contours around $\infty$ and $0$ respectively. 

It was proven in \cite{KZ21} that there exists a tau-function $\tau_n({\bf t})$ 
of the constrained Toda hierarchy, and the wave function can be consistently 
expressed through
the tau-function as follows:
\beq\label{t10}
\psi (n, {\bf t}, z)=z^ne^{\xi ({\bf t}, z)}G_n^{1/2}({\bf t}, z)
\frac{\tau_n({\bf t}-[z^{-1}])}{\tau_n({\bf t})}, \quad z\to \infty ,
\eeq
\beq\label{t11}
\psi (n, {\bf t}, z^{-1})=z^{-n}e^{-\xi ({\bf t}, z)}\tilde G_n^{1/2}({\bf t}, z)
\frac{\tau_{n+1}({\bf t})}{\tau_n ({\bf t})}\,
\frac{\tau_{n+1}({\bf t}+[z^{-1}])}{\tau_n({\bf t})}, \quad z\to \infty .
\eeq
Here
\beq\label{t12}
\xi ({\bf t}, z)=\sum_{k\geq 1}t_k z^k,
\eeq
\beq\label{t13}
\begin{array}{c}
{\bf t}\pm [z^{-1}]=\Bigl \{t_1\pm z^{-1}, t_2\pm \frac{1}{2}z^{-2}, 
t_3\pm \frac{1}{3}z^{-3}, \ldots \Bigr \},
\end{array}
\eeq
and
\beq\label{t14}
G_n({\bf t}, z)=1-z^{-2}\left (\frac{\tau_{n+1}({\bf t})
\tau_{n-1}({\bf t}-[z^{-1}])}{\tau_n ({\bf t})\tau_n({\bf t}-[z^{-1}])}\right )^2,
\eeq
\beq\label{t15}
\tilde G_n({\bf t}, z)=1-z^{-2}\left (\frac{\tau_{n}({\bf t})
\tau_{n+2}({\bf t}+[z^{-1}])}{\tau_{n+1} ({\bf t})\tau_{n+1}({\bf t}+[z^{-1}])}\right )^2.
\eeq

Substituting these expressions into (\ref{t9}), one obtains an 
integral equation\footnote{Unlike the Toda lattice hierarchy, this equation 
is not bilinear but has a more
complicated form with square roots of (\ref{t14}), (\ref{t15}).}
for the tau-function which is the generating equation of the 
hierarchy. In particular, setting ${\bf t}-{\bf t}'=[a^{-1}]+[b^{-1}]$, $n-n'=2$, 
the residue calculus yields:
\beq\label{t16}
\begin{array}{l}
\displaystyle{
\frac{a^2b}{a-b}
G_n^{1/2}({\bf t}, a)\tilde G_{n-2}^{1/2}({\bf t}\! -\! 
[a^{-1}]\! -\! [b^{-1}], a)
\frac{\tau_n({\bf t}\! -\! [a^{-1}])}{\tau_n({\bf t})}
\frac{\tau_{n-1}({\bf t}\! -\! [a^{-1}]\! -\! [b^{-1}])
\tau_{n-1}({\bf t}\! -\! [b^{-1}])}{\tau^2_{n-2}( {\bf t}\! -\! [a^{-1}]\! -\! [b^{-1}])}
}
\\ \\
\displaystyle{
-\frac{ab^2}{a\! -\! b}
G_n^{1/2}({\bf t}, b)\tilde G_{n-2}^{1/2}({\bf t}\! -\! [a^{-1}]\! -\! [b^{-1}], b)
\frac{\tau_n({\bf t}\! -\! [b^{-1}])}{\tau_n({\bf t})}
\frac{\tau_{n-1}({\bf t}\! -\! [a^{-1}]\! -\! [b^{-1}])
\tau_{n-1}({\bf t}\! -\! [a^{-1}])}{\tau^2_{n-2}( {\bf t}\! -\! [a^{-1}]\! -\! [b^{-1}])}
}
\\ \\
\phantom{aaaaaaaaaaa}\displaystyle{=
ab \frac{\tau^2_{n-1}({\bf t}\! -\! [a^{-1}]\! -\! [b^{-1}])}{\tau^2_{n-2},
{\bf t}\! -\! [a^{-1}]\! -\! [b^{-1}])}-(ab)^{-1}
\frac{\tau^2_{n+1}({\bf t})}{\tau_n^2({\bf t})}.}
\end{array}
\eeq

Finally, let us point out the relation between the tau-function of the constrained 
Toda hierarchy and the tau-function $\tau_n^{\rm Toda}
(\ldots t_{-2}, t_{-1}; t_1, t_2, \ldots )$ of the Toda lattice hierarchy satisfying the 
constraints
\beq\label{t17}
(\p_{t_k}+\p_{t_{-k}})\log \tau_n^{\rm Toda}\Bigr |_{t_j+t_{-j}=0}=0 \quad 
\mbox{for all $k,j$ and $t_0$, $t_j-t_{-j}$}
\eeq
which are equivalent to (\ref{t5}). We have (see \cite{KZ21}):
\beq\label{t18}
\tau_n({\bf t})=\sqrt{\vphantom{A^{\oint}}\tau_n^{\rm Toda}
(\ldots -t_2, -t_1; t_1, t_2, \ldots )}.
\eeq

\section{The dispersionless limit}

In the dispersionless limit \cite{TT91,TT95}, one should substitute
\beq\label{d1}
t_k \longrightarrow \frac{t_k}{\hbar}, \qquad \tau_{t_0/\hbar}({\bf t}/\hbar )=
e^{F(t_0, {\bf t}, \hbar )/\hbar^2}
\eeq
and take the $\hbar \to 0$ limit which allows one to obtain an equation for the function 
$F=\lim\limits_{\hbar \to 0}F(t_0, {\bf t}, \hbar )$ from the 
equation for the tau-function. In accordance with this prescription, we can write
$$
\tau_{n\pm 1}({\bf t})\longrightarrow \exp \Bigl (\hbar^{-2}e^{\pm \hbar \p_{t_0}}F\Bigr ),
$$
$$
\tau_n ({\bf t}\pm [a^{-1}])\longrightarrow \exp \Bigl (
\hbar^{-2}F(t_0, t\pm \hbar [a^{-1}], \hbar)\Bigr )=
\exp \Bigl (\hbar^{-2}e^{\pm \hbar D(a)}F\Bigr ),
$$
where $D(z)$ is the differential operator
\beq\label{d1a}
D(z)=\sum_{k\geq 1}\frac{z^{-k}}{k}\, \p_{t_k}.
\eeq

The dispersionless limit of equation (\ref{t16}) is then straightforward. 
We only note that
\beq\label{d2}
\lim_{\hbar \to 0}G_{t_0/\hbar}({\bf t}/\hbar , z)=\lim_{\hbar \to 0}
\tilde G_{t_0/\hbar}({\bf t}/\hbar , z)=
1-z^{-2}e^{2\p_{t_0}(\p_{t_0}+D(z))F},
\eeq
so the square roots in (\ref{t17}) disappear in the $\hbar \to 0$ limit. The final
result is
\beq\label{d3}
\frac{ae^{-\p_{t_0}(\p_{t_0}+D(a))F}\! -\! be^{-\p_{t_0}(\p_{t_0}+D(b))F}}{a-b}=
\frac{abe^{-\p_{t_0}(2\p_{t_0}+D(a)+D(b))F}\! -\! 1}{ab}\,
e^{(\p_{t_0}+D(a))(\p_{t_0}+D(b))F}.
\eeq
This is the dispersionless analogue of the bilinear 
equation for the tau-function. In the zero dispersion limit it is not bilinear
since it is written for logarithm of the tau-function. 
Differential equations of the hierarchy are obtained from (\ref{d3}) by
expanding both sides in inverse powers of $a$, $b$ and equating the
coefficients. For example, the simplest nontrivial equation is
\beq\label{d3a}
F_{01}^2-F_{02}+2F_{11}=2e^{F_{00}},
\eeq
where $F_{jk}:=\p_{t_j}\p_{t_k}F$.

Introducing the function 
\beq\label{d4}
w(z)=ze^{-\p_{t_0}(\p_{t_0}+D(z))F},
\eeq
we can write equation (\ref{d3}) in the form
\beq\label{d5}
\frac{w(a)-w(b)}{w(a)w(b)-1}=(b^{-1}-a^{-1}) e^{(\p_{t_0}+D(a))(\p_{t_0}+D(b))F}.
\eeq
Taking logarithm of both sides and applying the $t_0$-derivative, we can represent
it in the form of a closed equation for the function $w(z)$:
\beq\label{d6}
(\p_{t_0}+D(a))\log w(b)=-\p_{t_0}\log \frac{w(a)-w(b)}{w(a)w(b)-1}.
\eeq

It is instructive to rewrite this equation in terms of the inverse function to
$w(z)$ which we denote as $z(w)$. Noting the relation
$$
\p_{t_k}w(z)=-\frac{\p_{t_k}z(w)}{\p_w z(w)}
$$
between the partial derivatives, and, as a consequence,
$$
w\p_w z(w)\, D(a)\log w(z)=-D(a) z(w)
$$
(the derivative in the right hand side is taken at constant $w$), we can rewrite
equation (\ref{d6}) in the form
\beq\label{d7}
D(a)z(w)=\left \{ z(w), \, \log \frac{w(a)-w}{w(a)-w^{-1}}\right \},
\eeq
where
\beq\label{d8}
\{f,g\}=w\frac{\p f}{\p w}\frac{\p g}{\p t_0}-w\frac{\p f}{\p t_0}\frac{\p g}{\p w}
\eeq
is the Poisson bracket. Equation (\ref{d7}) is the generating Lax equation for the
dispersionless constrained Toda hierarchy. To see this, we introduce the Faber
polynomials $B_k(w)$ according to the expansion
\beq\label{d9}
-\log \Bigl (w(z)-w\Bigr )+\frac{1}{2}\log \frac{zw(z)}{r}=\sum_{k\geq 1}\frac{z^{-k}}{k}
\, B_k(w), \quad r=(\lim_{z\to \infty}(w(z)/z))^{-1}.
\eeq
Here $w$ is an independent variable ($w\neq w(z)$).
Then equation (\ref{d7}) implies the Lax equations of the dispersionless constrained 
Toda hierarchy:
\beq\label{d10}
\p_{t_k}z(w)=\Bigl \{ B_k(w)-B_k(w^{-1}), \, z(w)\Bigr \}.
\eeq
It is easy to see that
\beq\label{d11}
\begin{array}{c}
B_k(w)=\Bigl (z^k(w)\Bigr )_{>0}+\frac{1}{2} \Bigl (z^k(w)\Bigr )_{0}.
\end{array}
\eeq
Comparing with (\ref{t1}), (\ref{t6}), we see that the function (or rather the
Laurent series) $z(w)$ is the dispersionless limit of the Lax operator:
$e^{\p_n}\to w$, $L \to z(w)$ and the formulas (\ref{d10}), (\ref{d11}) are
dispersionless limits of (\ref{t7}), (\ref{t2}). 

From now on, keeping in mind applications to conformal maps of planar domains,
we we will consider the real form of the Toda lattice hierarchy such that 
the time variables satisfy the conditions
\beq\label{d20}
t_{-k}=-\bar t_k, \quad k\geq 1, \quad \mbox{and $t_0$ is real}.
\eeq
Moreover, the tau-function of the real form of the Toda lattice hierarchy is real. 
For the constrained Toda hierarchy this means that $t_k=\bar t_k$, i.e., all 
the time variables are real. 

The dispersionless limit of the Lax operator $L$ is the Laurent series 
$z(w)$ while the limit of the 
second Lax operator, $\bar L$, is $\bar z(w^{-1})=\overline{z(\bar w^{-1})}$.
Since $e^{\p_n}$ in the dispersionless limit becomes $w$ and 
$(e^{\p_n})^{\dag}=e^{-\p_n}$, 
we have $w^{\dag}=w^{-1}$, so
the constraint (\ref{t5}) in the dispersionless limit
implies that $z(w)=\overline{z(\bar w)}$, or $w(z)=\overline{w(\bar z)}$, i.e.
all coefficients of the Laurent series $z(w)$ and $w(z)$ are real. 
The constraints (\ref{t17}) for the tau-function then mean that
$v_k=\p_{t_k}F$ are also real. 

\section{Dispersionless \ constrained Toda \
hierarchy and con\-for\-mal maps}

It is known that the real reduction of the dispersionless Toda
lattice hierarchy governs conformal maps of planar domains \cite{WZ00,MWZ02,Z01}.
Let us recall the main facts related to this connection. 

Given a simply-connected compact 
planar domain ${\sf D}$ with a smooth boundary, one can 
define the harmonic moments
\beq\label{h1}
v_k = \frac{1}{\pi}\int_{{\sf D}}z^k d^2z, \quad k\geq 1, \quad d^2z \equiv dxdy
\eeq
together with the complimentary moments of the exterior domain $\CC \setminus {\sf D}$
\beq\label{h2}
t_k=-\frac{1}{\pi k}\int_{\ccc \setminus {\sf D}}z^{-k}d^2z.
\eeq
Without loss of generality we assume that the domain ${\sf D}$ contains the origin.
The moments $t_k$ and $v_k$ are in general complex numbers. 
We also denote
\beq\label{h3}
t_0=\frac{1}{\pi}\int_{\sf D}d^2 z = \frac{\mbox{Area of ${\sf D}$}}{\pi}
\eeq
which is a real number. 

Let $z(w)$ be the conformal map from the exterior of the unit circle to the 
domain $\CC \setminus {\sf D}$ normalized in such a way that $z(\infty )=\infty$ and
$z'(\infty )=r>0$ ($r$ is a real positive number called conformal radius). 
As it was shown in \cite{WZ00}, the function $z(w)$ satisfies the equations
\beq\label{h4}
\p_{t_k}z(w)=\{ B_k(w), z(w)\}, \quad \p_{\bar t_k}z(w)=-\{\bar B_k (w^{-1}), z(w)\}
\eeq
with $B_k(w)$ given by (\ref{d11}). This allows one to identify $z(w)$ with the Lax
function of the dispersionless Toda hierarchy (the dispersionless limit of the Lax
operator). 

The dispersionless limit of logarithm of tau-function (the $F$-function) is given by
\beq\label{h5}
F^{\rm Toda}
=-\frac{1}{\pi^2}\int_{\sf D}\int_{\sf D}d^2z d^2\zeta \log \Bigl |z^{-1}-\zeta^{-1}\Bigr |
\eeq
(see \cite{MWZ02}).
It is a real-valued function of $t_0$, $t_k$ and $\bar t_k$. The first derivatives 
of $F$ with respect to $t_k$'s yield the complimentary moments $v_k$:
\beq\label{h6}
v_k=\p_{t_k}F^{\rm Toda}.
\eeq
The second mixed derivatives of $F$ provide full information about the conformal map
$w(z)$ from $\CC \setminus {\sf D}$ to the exterior of the unit circle (the inverse
function of $z(w)$) and the Green function $G(z,\zeta )$ 
of the Dirichlet boundary problem in 
$\CC \setminus {\sf D}$. Namely,
\beq\label{h7}
\begin{array}{c}
w(z)=z\exp \Bigl (-\frac{1}{2}\p_{t_0}^2 F^{\rm Toda} -\p_{t_0}D(z)F^{\rm Toda}\Bigr ),
\end{array}
\eeq
\beq\label{h8}
G(z, \zeta )=\log \Bigl |z^{-1}-\zeta^{-1}\Bigr |+\frac{1}{2}
\nabla (z)\nabla (\zeta )F^{\rm Toda},
\eeq
where
$$
\nabla (z)=\p_{t_0}+D(z) +\overline{D(z)}
$$
and $D(z)$ is the differential operator (\ref{d1a}). The well-known formula
\beq\label{h9}
G(z, \zeta )=\log \left |\frac{w(z)-w(\zeta )}{w(z)\overline{w(\zeta )}-1}\right |
\eeq
which expresses the Green function through the conformal map is equivalent to
the dispersionless limit of the Toda equations for the tau-function. 

In the case of the constrained Toda hierarchy, the coefficients of the Laurent series
$z(w)$ and the moments $t_k$ are real. It is obvious
from (\ref{h2}) that for domains 
symmetric under reflection with respect to the real axis
(reflection-symmetric domains) all moments are real. The converse
statement, i.e. that for real moments the corresponding domain is
symmetric, can be proved by means of the Schwarz function\footnote{The
Schwarz function $S(z)$ \cite{Davis} 
is a holomorphic function in a strip-like neighborhood
of the boundary of the domain such that $\bar z =S(z)$ on the boundary.
It is a generating function of the harmonic moments. Therefore, if all
harmonic moments are real, $\bar z=S(z)$ implies $z=S(\bar z)$, which
means that complex conjugation preserves the boundary.}.
This means that we deal with conformal maps
of reflection-symmetric domains. The moments $v_k$ are also real, and formula (\ref{h6})
means that the constraint (\ref{t17}) is satisfied  if we are 
in the class of reflection-symmetric domains. 
From (\ref{t18}) it follows that the
$F$-function of the dispersionless constrained Toda hierarchy is related to
$F^{\rm Toda}$ as
\beq\label{h10}
F^{\rm Toda}=2F.
\eeq
Therefore, one can see that formula (\ref{d4}) is a specialization of (\ref{h7})
for reflection-symmetric domains. In a similar way, equation (\ref{d5}) is a specialization
of (\ref{h8}) for reflection-symmetric 
domains with $z, \zeta$ lying on the real axis. 
In the class of reflection-symmetric domains, the function $F$, 
regarded as a function of
real moments $t_k$, satisfies equation (\ref{d3}). 

As it follows from the result of \cite{MWZ02}, the function $F$ satisfies
the following quasi-homogeneity equation:
\beq\label{h11}
2F=-\frac{1}{4}\, t_0^2 +t_0\p_{t_0}F +\sum_{k\geq 1}(2\! -\! k)t_k \p_{t_k}F.
\eeq

\noindent
{\bf Example.} For the ellipse with half-axes $a,b$ centered at $x_0\in \RR$ we have:
$$
t_0=ab , \quad t_1=\frac{2bx_0}{a+b}, \quad 2t_2=\frac{a-b}{a+b},
$$
and all other moments are equal to zero. For the ellipse the function $F$ is
\beq\label{h12}
F=\frac{1}{4}\, t_0^2\log t_0 -\frac{3}{8}\, t_0^2 -\frac{1}{4}\, t_0^2
\log (1-4t_2^2)+\frac{t_0t_1^2}{2(1-2t_2)}
\eeq
(see \cite{WZ00}). One can check that this function does satisfy 
equation (\ref{d3a}). 
 
\section{One-variable reductions}

In this section we consider reductions of the dispersionless 
constrained Toda hierarchy. For one-variable reductions, the dependence of 
$w(z)$ on the times is implemented by means of a single variable
$\lambda =\lambda ({\bf t})$, i.e., $w(z;{\bf t})=w(z, \lambda ({\bf t}))$.
This means that instead of the function of infinitely many variables
$w(z;{\bf t})$ we now deal with a function of two variables $w(z, \lambda )$. 
Our goal is to find all possible forms of this function consistent with the
infinite hierarchy. The reduction is an exceptional non-generic solution. 
We will see that such solutions correspond to conformal maps of domains with 
slits. 

The derivation below is parallel to that for reductions of the dispersionless
Toda hierarchy \cite{ATZ21} (see also \cite{TTZ06,Z07}). 
Assuming that $w(z)=w(z, \lambda )$ and using the chain rule of differentiating, we have:
$$
(\p_{t_0}+D(\zeta ))\log w(z)=\p_{\lambda }\log w(z)
\cdot (\p_{t_0}+D(\zeta ))\lambda 
$$
and
\beq\label{or0}
(\p_{t_0}+D(z))\lambda =
\frac{(\p_{t_0}+D(z))\log r}{\p_{\lambda}\log r}.
\eeq
We recall that $\log r =-\log \lim\limits_{z\to \infty}(w(z)/z)$.
Tending $b\to \infty$ in equation (\ref{d6}), we get:
\beq\label{or1}
\p_{t_0}\log w(z)=-(\p_{t_0}+D(z))\log r.
\eeq
Therefore, equation (\ref{or0}) can be written as
\beq\label{or2}
(\p_{t_0}+D(z))\lambda =-
\frac{\p_{\lambda} \log w(z)}{\p_{\lambda}\log r}\, 
\p_{t_0}\lambda .
\eeq
Substituting this into (\ref{d6})
and assuming that $\p_{t_0}\lambda$ is not identically zero, we get:
$$
\p_{\lambda }\log w(z)\p_{\lambda }\log w(\zeta )=\p_{\lambda}\log r
\left [\p_{\lambda}w(z)\Bigl (\frac{1}{w(z)-w(\zeta )}-\frac{w(\zeta )}{w(z)w(\zeta )-1}
\Bigr )\right.
$$
$$
\left. -\p_{\lambda}w(\zeta )\Bigl (\frac{1}{w(z)-w(\zeta )}+\frac{w(z )}{w(z)w(\zeta )-1}
\Bigr )\right ]
$$
or, after rearranging,
$$
w(z)+w^{-1}(z)+\p_{\lambda}\log r \frac{w(z)-w^{-1}(z)}{\p_{\lambda}\log w(z)}=
w(\zeta )+w^{-1}(\zeta )+\p_{\lambda}\log r 
\frac{w(\zeta )-w^{-1}(\zeta )}{\p_{\lambda}\log w(\zeta )}.
$$
The left hand side depends only on $z$ while the right hand side depends only on $\zeta$.
It then follows from this equation that
\beq\label{or3}
\eta (\lambda ):= 
w(z)+w^{-1}(z)+\p_{\lambda}\log r \frac{w(z)-w^{-1}(z)}{\p_{\lambda}\log w(z)}
\eeq
does not depend on $z$. Equation (\ref{or3}) can be read as the differential equation
for the function $w(z)=w(z, \lambda )$:
\beq\label{or4}
\p_{\lambda}\log w(z)=-\frac{w(z)-w^{-1}(z)}{w(z)+w^{-1}(z)-\eta (\lambda )}
\, \p_{\lambda}\log r,
\eeq
where $\eta (\lambda )$ is a real-valued function of $\lambda$. 
Setting
\beq\label{or5}
\eta (\lambda )=e^{i\xi (\lambda )}+e^{-i\xi (\lambda )},
\eeq
where $\xi (\lambda )$ is another real-valued function of $\lambda$, we can
represent equation (\ref{or4}) in the form
\beq\label{or6}
\p_{\lambda}\log w(z)=-\frac{1}{2}
\left (\frac{w(z)+e^{i\xi (\lambda )}}{w(z)-e^{i\xi (\lambda )}}+
\frac{w(z)+e^{-i\xi (\lambda )}}{w(z)-e^{-i\xi (\lambda )}}\right )
\p_{\lambda}\log r.
\eeq
In particular, one can choose $\lambda =\log r$, then the equation simplifies:
\beq\label{or7}
\frac{\p \log w(z)}{\p \log r}=-\frac{1}{2}
\left (\frac{w(z)+e^{i\xi }}{w(z)-e^{i\xi }}+
\frac{w(z)+e^{-i\xi }}{w(z)-e^{-i\xi }}\right ).
\eeq

We call equation (\ref{or6}) (or (\ref{or7})) {\it symmetric radial L\"owner equation}.
Its right hand side is the half-sum of the right hands sides of the radial 
L\"owner equations with the driving functions $\xi (\lambda )$ and $-\xi (\lambda )$. 
This is similar to the quadrant L\"owner equation \cite{T13} 
which appears as a reduction condition
in the dispersionless BKP (or CKP) hierarchy\footnote{As is shown in 
\cite{Z21}, the dispersionless limits of the BKP and CKP hierarchies are
the same.}
and whose right hand side is the 
half-sum of the right hands sides of the chordal 
L\"owner equations with the driving functions 
$\xi (\lambda )$ and $-\xi (\lambda )$. Moreover, in a certain scaling limit
the symmetric radial L\"owner equation becomes the quadrant L\"owner equation.
To wit, setting $w(z)=e^{i\epsilon p(z)}$, $r=e^{-\epsilon^2 u/2}$,
$\xi (\lambda )\to \epsilon \xi (\lambda )$ and 
taking the limit $\epsilon \to 0$, we obtain the quadrant L\"owner equation
\beq\label{quad}
\p_{\lambda}p(z)=-\frac{1}{2}\left ( \frac{1}{p(z)-\xi (\lambda )}
+\frac{1}{p(z)+\xi (\lambda )}\right )\p_{\lambda}u.
\eeq
As is shown in Section \ref{section:symmetric}, 
solutions of the symmetric radial L\"owner
equation are functions which define conformal mappings from the exterior
of the unit circle with two symmetric slits to the exterior of the unit circle.

The dependence of $\lambda$ on the times $t_k$ is determined by a system of
equations of the hydrodynamic type. They follow from equation (\ref{or1}) 
which can be written as
$$
(\p_{t_0}+D(z))\lambda =\frac{(\p_{t_0}+D(z))\log r}{\p_{\lambda}\log r}=
-\frac{\p_{t_0}\log w(z)}{\p_{\lambda}\log r}=-
\frac{\p_{\lambda}\log w(z)}{\p_{\lambda}\log r}\, \p_{t_0}\lambda .
$$
Plugging here the symmetric L\"owner equation (\ref{or6}), we obtain
\beq\label{or8}
D(z)\lambda =\left (\frac{e^{i\xi (\lambda )}}{w(z)-e^{i\xi (\lambda )}}+
\frac{e^{-i\xi (\lambda )}}{w(z)-e^{-i\xi (\lambda )}}\right ) \p_{t_0}\lambda .
\eeq
Taking the $w$-derivative of equation (\ref{d9}) defining the Faber polynomials,
we get
\beq\label{or9}
\frac{e^{\pm i\xi (\lambda )}}{w(z)-e^{\pm i\xi (\lambda )}}=
\sum_{k\geq 1}\frac{z^{-k}}{k}\, \phi_k(e^{\pm i\xi} ), \quad
\phi_k(w)\equiv w\p_w B_k(w).
\eeq
Therefore, equation (\ref{or8}) is equivalent to the following infinite
system of equations of hydrodynamic type:
\beq\label{or10}
\p_{t_k}\lambda =\Bigl (\phi_k(e^{i\xi (\lambda )})+\phi_k(e^{-i\xi (\lambda )})
\Bigr )\p_{t_0}\lambda .
\eeq
The hodograph equation
\beq\label{or11}
t_0+\sum_{k\geq 1}t_k \Bigl (\phi_k(e^{i\xi (\lambda )})+\phi_k(e^{-i\xi (\lambda )})
\Bigr )
=R(\lambda )
\eeq
gives a general solution in an implicit form. Here $R(\lambda )$ is an 
arbitrary function of $\lambda$. 

\section{Multivariable reductions}

A multivariable
($N$-variable) reduction of the dispersionless constrained Toda hierarchy
means that the function $w(z)$ depends on the times ${\bf t}$ 
through $N$ functions $\lambda_j=\lambda_j({\bf t})$, i.e., 
$w(z)=w(z; {\bf t})=w(z; \lambda_1({\bf t}), \ldots , \lambda_N({\bf t}))$. 
Below we prove that solutions of a system of symmetric radial L\"owner 
equations give solutions to the hierarchy.

We consider the system of $N$ symmetric radial L\"owner equations of the form 
(\ref{or6}) which 
characterize the dependence of $w(z)=w(z;\lambda_1, \ldots , \lambda_N)$ 
on the variables $\lambda_j$:
\beq\label{m1}
\frac{\p \log w(z)}{\p \lambda_j}=-\frac{1}{2}
\left (\frac{w(z)+\nu_j}{w(z)-\nu_j}+
\frac{w(z)+\bar \nu_j}{w(z)-\bar \nu_j}\right )
\frac{\p \log r}{\p \lambda_j}, \quad \nu_j :=e^{i\xi_j}.
\eeq
The driving functions $\xi_j = \xi_j (\{\lambda_i\})$ are real-valued, so
$|\nu_j|=1$, i.e., $\bar \nu_j =\nu_j^{-1}$. 
The compatibility conditions of this system are
$$\displaystyle{\frac{\p }{\p \lambda_k}\frac{\p \log w(z)}{\p \lambda_j}=
\frac{\p }{\p \lambda_j}\frac{\p \log w(z)}{\p \lambda_k}}.$$ 
In order to resolve these conditions, one should substitute 
equations (\ref{m1}) and cancel all poles.
A long but straightforward calculation shows that the compatibility
conditions are equivalent to
the following system of the Gibbons-Tsarev type:
\beq\label{m2}
\left \{
\begin{array}{l}
\displaystyle{
\frac{\p \nu_j}{\p \lambda_k}=\frac{\nu_j}{2}\left (
\frac{\nu_k+\nu_j}{\nu_k-\nu_j}+\frac{\bar \nu_k+\nu_j}{\bar \nu_k-\nu_j}\right )
\frac{\p \log r}{\p \lambda_k},}
\\ \\
\displaystyle{
\frac{\p^2 \log r}{\p \lambda_j \p \lambda_k}=
2\left (\frac{\nu_j \nu_k}{(\nu_j-\nu_k)^2}+\frac{\nu_j \bar \nu_k}{(\nu_j-\bar \nu_k)^2}
\right ) \frac{\p \log r}{\p \lambda_j}\, \frac{\p \log r}{\p \lambda_k}.}
\end{array}
\right.
\eeq

Now we are going to use the fact that each $\lambda_j$
is a function of the times:
$w(z;{\bf t})=w(z;\{\lambda_j({\bf t})\})$.
Applying the chain rule of differentiating, in the case of $N$-variable reduction
we can rewrite equation (\ref{d6}) of the dispersionless constrained Toda
hierarchy in the form
\beq\label{m3}
\sum_{j=1}^N D(\zeta )\lambda_j \cdot \p_{\lambda_j}\log w(z)=\sum_{j=1}^N
\p_{t_0}\lambda_j \cdot 
\p_{\lambda_j}\log \left (\frac{w(\zeta )-w^{-1}(z)}{w(\zeta )-w(z)}\right ).
\eeq
A rather long but straightforward calculation using the symmetric L\"owner
equation (\ref{m1}) yields
\[
 \begin{split}
    &\p_{\lambda_j}\log 
    \left (\frac{w(\zeta )-w^{-1}(z)}{w(\zeta )-w(z)}\right )
\\
    &=
    -\frac{1}{2}
    \left (
      \frac{\nu_j}{w(\zeta )-\nu_j}
     +\frac{\bar \nu_j}{w(\zeta )-\bar \nu_j}
    \right) 
    \left (
       \frac{w(z)+\nu_j}{w(z)-\nu_j}
      +\frac{w(z)+\bar \nu_j}{w(z)-\bar \nu_j}
    \right )
    \p_{\lambda_j}\log r.
\\
    &=
    \left (
      \frac{\nu_j}{w(\zeta )-\nu_j}
     +\frac{\bar \nu_j}{w(\zeta )-\bar \nu_j}
    \right) 
    \frac{\p \log w(z)}{\p \lambda_j}.
 \end{split}
\]
Therefore, if we introduce the dependence of the $\lambda_j$'s on the times by
means of the relations
\beq\label{m4}
D(z)\lambda_j =\left (\frac{\nu_j}{w(z)-\nu_j}+
\frac{\bar \nu_j}{w(z)-\bar \nu_j}\right ) \p_{t_0}\lambda_j,
\eeq
the equation (\ref{m3}) will be satisfied identically. 

As it follows from (\ref{or9}), equation (\ref{m4}) is equivalent to an infinite system
of partial differential equations of hydrodynamic type:
\beq\label{m5}
\frac{\p \lambda_j}{\p t_k}=\Bigl (\phi_{j,k}(\{\lambda_i\})+
\overline{ \phi_{j,k}(\{\lambda_i\})}\, \Bigr ) \frac{\p \lambda_j}{\p t_0}, 
\quad \phi_{j,k}=\nu_j B'_k (\nu_j), 
\eeq
where $B'_k(w)\equiv \p_w B_k(w)$. The generating function of $\phi_{j,k}$'s is
\beq\label{m6}
\begin{array}{c}
\displaystyle{
\sum_{k\geq 1}\Bigl (\phi_{j,k}(\{\lambda_i\})+
\overline{ \phi_{j,k}(\{\lambda_i\})}\, \Bigr )\frac{z^{-k}}{k}=
\frac{\nu_j}{w(z)-\nu_j}+
\frac{\bar \nu_j}{w(z)-\bar \nu_j}}
\\ \\
\displaystyle{
:=Q\Bigl (w(z, \{\lambda_i\}), \nu_j (\{\lambda_i\})\Bigr ).}
\end{array}
\eeq

Now we are going to show that the system (\ref{m5}) is consistent and can be solved by Tsarev's
generalized hodograph method \cite{Tsa90}. It can be directly verified that
the compatibility condition of the system (\ref{m5}) is
\beq\label{m7}
\frac{\p_{\lambda_j} {\rm Re}\, \phi_{i,n}}{{\rm Re}\,\phi_{j,n}-{\rm Re}\,\phi_{i,n}}=
\frac{\p_{\lambda_j} {\rm Re}\, \phi_{i,n'}}{{\rm Re}\, \phi_{j,n'}-{\rm Re}\,\phi_{i,n'}}
\quad \mbox{for all $i\neq j, \, n, \, n'$}.
\eeq
In other words, the condition is
that 
\beq\label{m8}
\Gamma_{ij}:=\frac{\p_{\lambda_j} {\rm Re}\, 
\phi_{i,n}}{{\rm Re}\,\phi_{j,n}-{\rm Re}\,\phi_{i,n}}
\eeq
does not depend on $n$. It is easy to see that this is equivalent to 
the $z$-independence of the ratio
$$
\frac{\p_{\lambda_j}Q(w(z), \nu_i)}{Q(w(z), \nu_j)-Q(w(z), \nu_i)},
$$
where $Q$ is the generating function (\ref{m6}). If this holds, then
\beq\label{m9}
\Gamma_{ij}=\frac{\p_{\lambda_j}Q(w(z), \nu_i)}{Q(w(z), \nu_j)-Q(w(z), \nu_i)}.
\eeq
A direct calculation which makes use of the symmetric L\"owner equation (\ref{m1})
and the Gibbons-Tsarev equations (\ref{m2}) gives
\beq\label{m10}
\Gamma_{ij}=\left (\frac{\nu_i\nu_j}{(\nu_i-\nu_j)^2}+
\frac{\nu_i\nu_j}{(\nu_i\nu_j-1)^2}\right ) \p_{\lambda_j}\log r .
\eeq

Let $R_i=R_i(\{\lambda _j\})$
($i=1,\ldots ,N$) satisfy the system of equations  
\beq\label{m11}
\frac{\p R_i}{\p \lambda_j}=\Gamma_{ij}(R_j-R_i), \qquad
i,j=1, \ldots , N, \quad i\neq j,
\eeq
where $\Gamma_{ij}$ is defined in (\ref{m10}) (for $N=1$ this condition is void). We
claim that
the system (\ref{m11}) is compatible in the sense of Tsarev \cite{Tsa90}.
To see this, we note that $\Gamma_{ij}$ 
can be expressed as
\beq\label{m12}
\Gamma_{ij}=\frac{1}{2}\, \p_{\lambda_j}\log g_i, \quad
g_i=\frac{\p \log r}{\p \lambda_i}.
\eeq
This directly follows from the Gibbons-Tsarev equations (\ref{m2}).
It is then obvious that 
\beq\label{m13}
\frac{\p \Gamma_{ij}}{\p \lambda_k}=\frac{\p \Gamma_{ik}}{\p \lambda_j}, \qquad i\neq j\neq k,
\eeq
which are Tsarev's compatibility conditions. 
This means that the system (\ref{m5}) is semi-Hamiltonian.
The main geometric object associated with a semi-Hamiltonian system is a diagonal metric.
The quantities $g_i=g_{ii}$ are components of this metric while $\Gamma_{ij}=\Gamma_{ij}^{i}$ 
are the corresponding Christoffel symbols.
Moreover, from (\ref{m12}) it is clear that the metric $g_i$ is of Egorov type, i.e.,
\beq\label{m14}
\frac{\p g_i}{\p \lambda_k}=\frac{\p g_k}{\p \lambda_i}.
\eeq

Assume that $R_i$ satisfy the system (\ref{m11}). 
Then the same argument as 
in the proof of Theorem 10 of Tsarev's paper \cite{Tsa90} shows that
if $\lambda_i$ is defined implicitly by the
hodograph relations
\beq\label{m15}
t_0+2{\rm Re}\sum_{k\geq 1}t_k \phi_{i,k}(\{\lambda_j \})
=R_i(\{\lambda_j\} ),
\eeq
then $\lambda_j$ satisfy (\ref{m5}). 

We have found sufficient conditions for $N$-variable 
diagonal reductions of the dispersionless constrained Toda hierarchy. 
The reduction is given by a system of $N$ symmetric
L\"owner equations (\ref{m1}) for a function $w(z, \lambda_1, \ldots , \lambda_N)$
supplemented by a diagonal system of hydrodynamic type (\ref{m5}) 
for the variables 
$\lambda_j$, $j=1, \ldots , N$.

\section{The symmetric radial L\"owner equation}

\label{section:symmetric}

In this section we clarify the geometric meaning of the symmetric radial
L\"owner equation (\ref{or6}). We will sketch the derivation of a
differential equation obeyed by the conformal map of the unit disk with
two symmetric slits to the exterior of the unit disk which turns out to
be the symmetric radial L\"owner equation. Rigorous proof is in
\appref{app:proof-rad-sym-loewner}.

By $\UU$ we denote the unit disk $|z|<1$ and $\UU ^c =\CC \setminus \UU$
its compliment.  Let $\Gamma : [0, \infty )\to \UU ^c$ be 
%a smooth curve
a simple curve
with no self-intersections in the upper half plane with the condition
$|\Gamma (0)|=1$ and $\bar \Gamma$ be the mirror symmetric curve with
respect to the real axis: $\bar \Gamma (t )=\overline{\Gamma (t )}$. The
curves $\Gamma$ and $\bar \Gamma$ do not intersect, 
as $\Gamma(t)$ is in the upper half plane. 
Let $\Gamma_{t}$ be
the arc of the curve $\Gamma$ corresponding to the values of the
parameter from $0$ to $t$: $\Gamma_t : [0, t]\to \UU ^c$, and $\bar
\Gamma_t$ be the symmetric arc in the lower half plane.

Let $g(z, t)$ be the univalent conformal map from the domain
$\DD_t =$ (interior of $\UU^c$) $\setminus (\Gamma_t \cup \bar
\Gamma_t)$ to $\interior{\UU^c} = $ (interior of $\UU^c$) normalized by
the condition
\begin{equation}
\label{sl1} 
    g(z,t)=z/r(t) +O(1), \quad z\to \infty , \quad r(t)\in \RR_{+}.  
\end{equation}
The Riemann mapping theorem ensures that such a map exists and is
unique. The quantity $r(t)$ is called the conformal radius of the domain
$\DD_t$
relative to infinity. 

\begin{figure}[h]
 \centering
 \def\svgwidth{.9\linewidth}
 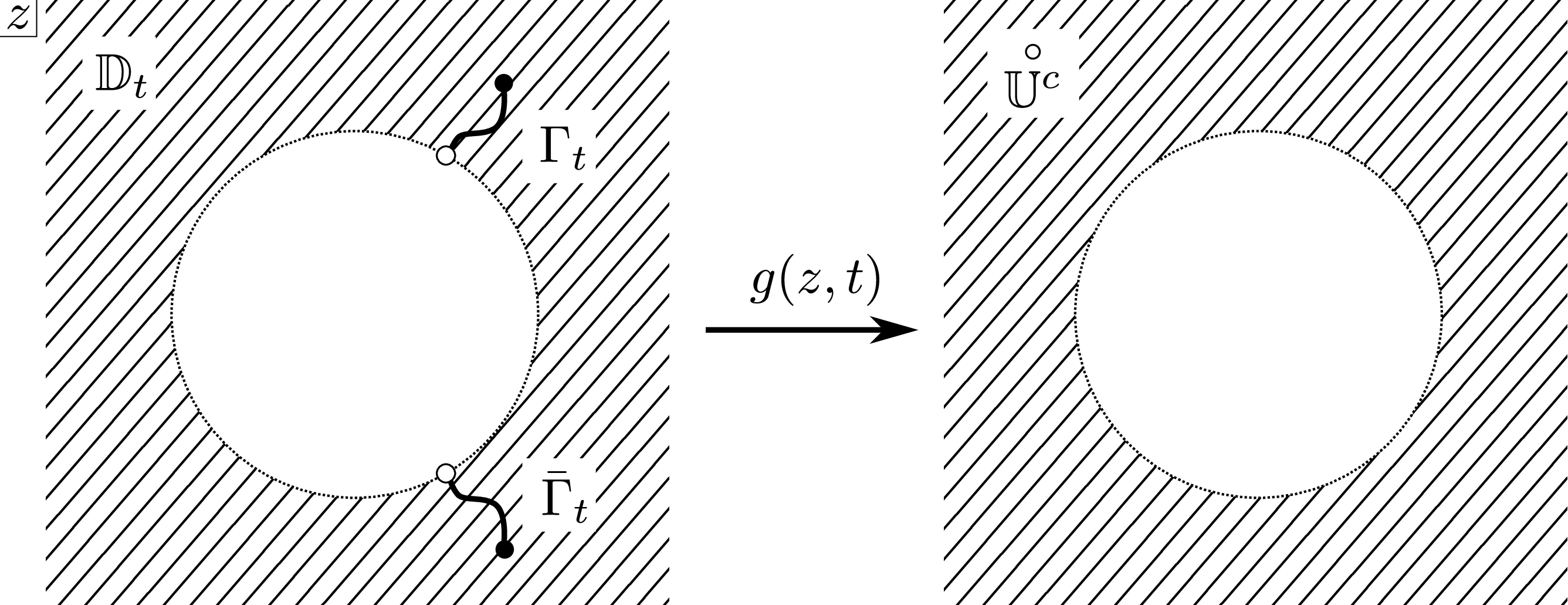
 \caption{Conformal map $g(z,t)$ from a symmetric slit domain to the
 outside of the unit disk.}
\label{fig:g(z,t)}
\end{figure}

Our goal is to derive a differential equation for the function $g(z,t)$
in the variable $t$.  According to our convention, $g(z,0)=z$, which is
going to be the initial condition for the differential equation. Note
that since the domain $\DD_t$ is symmetric
and $g(z,t)$ is normalised as \eqref{sl1} ($r(t)\in\Real_+$), 
the function $g(z, t)$ enjoys the property
\begin{equation}
\label{sym}
    g(z,t)=\overline{g(\bar z, t)}. 
\end{equation}
(See \lemref{lem:reflection}.)

The main technical tool is the complex Poisson formula (or the Schwarz
integral formula) which is a version of the Cauchy integral
formula. Given
a bounded continuous function $f(z)$ in $\UU^c$, holomorphic in its
interior $\interior{\UU^c}$ (thus holomorphic also at $z=\infty$ by
Riemann's removable singularity theorem), the complex Poisson formula
reads
\beq\label{sl2}
    f(z)
    =
    \frac{1}{2\pi}\int_{-\pi}^{\pi} 
     {\rm Re}\, f(e^{i\theta}) \frac{z+e^{i\theta}}{z-e^{i\theta}}\, 
    d\theta 
    +
    i\,{\rm Im}\, f(\infty ).
\eeq
We define a map $h(z;s,t)$ ($0<s<t$) for $z\in \UU^c$ by
\beq\label{sl3}
h(z; s, t)=g(g^{-1}(z,t), s)=\frac{r(t)}{r(s)}\, z +O(1),
\eeq
where $g^{-1}$ is the function inverse to $g$. (See
\figref{fig:h(z;s,t)}.) 

\begin{figure}[h]
 \centering
 \def\svgwidth{.8\linewidth}
 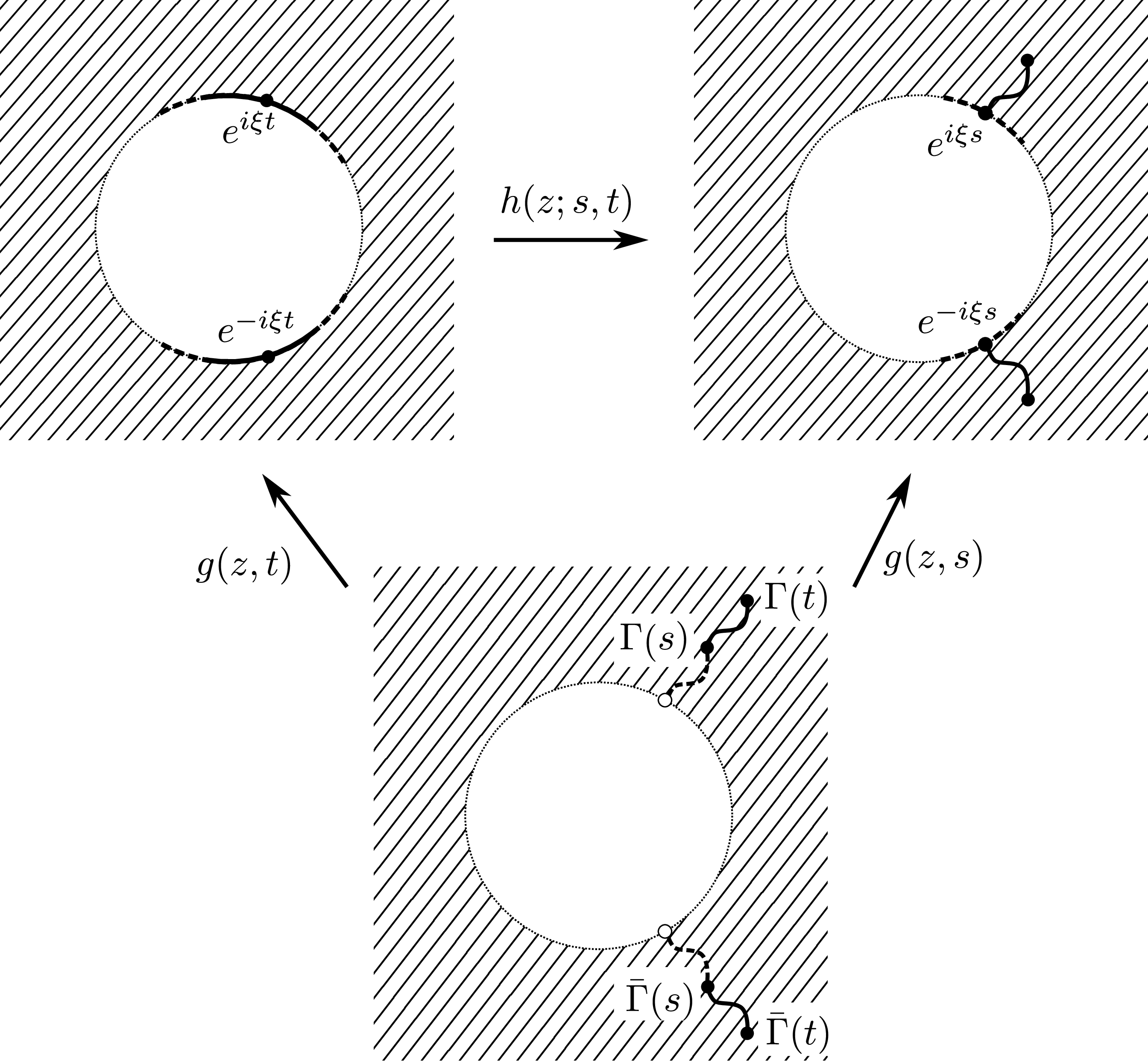
 \caption{Conformal map $h(z;s,t)=g(g^{-1}(z,t), s)$ mapping the outside
 of the unit disk to a symmetric slit domain.}
\label{fig:h(z;s,t)}
\end{figure}

We want to apply the complex Poisson formula to the function
$\log(h(z;s,t)/z)$.  Hence, hereafter we {\em assume} that $g^{-1}(z,t)$
and hence $h(z;s,t)$ are continuously extended to $\UU^c$, the closure
of the original domain of definition $\interior{\UU^c}$. (In
\appref{app:proof-rad-sym-loewner} we shall see that such extension is
really possible for certain class of functions.)

We get:
\beq\label{sl4}
\log \frac{h(z;s,t)}{z}=\frac{1}{2\pi}
\int_{-\pi}^{\pi} \log |h(e^{i\theta}); s,t)|
\, \frac{z+e^{i\theta}}{z-e^{i\theta}}\, d\theta .
\eeq
In particular, tending $z\to \infty$, we obtain the following corollary
of (\ref{sl4}):
\beq\label{sl6}
\log r(t)-\log r(s)=\frac{1}{2\pi}
\int_{-\pi}^{\pi} \log |h(e^{i\theta}); s,t)|d\theta .
\eeq
Substituting $z\mapsto g(z, t)$ in (\ref{sl4}), we can write:
\beq\label{sl5}
\log \frac{g(z,s)}{g(z, t)}=\frac{1}{2\pi}
\int_{-\pi}^{\pi} \log |h(e^{i\theta}); s,t)|
\, \frac{g(z, t)+e^{i\theta}}{g(z, t)-e^{i\theta}}\, d\theta .
\eeq
The property (\ref{sym}) implies the similar property of the function
$h(z; s,t)$: $h(z; s,t)=\overline{h(\bar z; s,t)}$ which allows one to
rewrite the relations (\ref{sl6}), (\ref{sl5}) as follows:
\beq\label{sl6a}
\log r(t)-\log r(s)=\frac{1}{\pi}
\int_{0}^{\pi} \log |h(e^{i\theta}); s,t)|d\theta ,
\eeq
\beq\label{sl5a}
\log \frac{g(z,s)}{g(z, t)}=\frac{1}{2\pi}
\int_{0}^{\pi} \log |h(e^{i\theta}); s,t)|
\left (\frac{g(z, t)+e^{i\theta}}{g(z, t)-e^{i\theta}}
+\frac{g(z, t)+e^{-i\theta}}{g(z, t)-e^{-i\theta}}\right )
d\theta .
\eeq

Let the image of the tip of the curve $\Gamma_t$ be $e^{i\xi (t)}$:
$g(\Gamma (t), t)=e^{i\xi (t)}$. Then the image of the tip of the curve
$\bar \Gamma_t$ is $e^{-i\xi (t)}$. Note that $\log |h(e^{i\theta};
s,t)|=0$ if $e^{i\theta}$ is mapped by $h$ to the boundary of $\UU^c$,
i.e. to a point on the unit circle. Accordingly, $\log |h(e^{i\theta};
s,t)|\neq 0$ if the point $e^{i\theta}$ lies in a neighborhood of
$e^{i\xi (t)}$ (in an arc of the unit circle containing the point
$e^{i\xi (t)}$) or in a neighborhood of $e^{-i\xi (t)}$. If $s\to t$
from below, this neighborhood shrinks to a point ($e^{i\xi (t)}$ or
$e^{-i\xi (t)}$).  Therefore, equations (\ref{sl6a}), (\ref{sl5a}) imply
the following differential equation by the mean value
theorem\footnote{Of course in this argument derivative $\der_t$ in
\eqref{sl7} should be interpreted as the left derivative, but we can
show the same equation for the right derivative, too. See
\appref{app:proof-rad-sym-loewner}.}:
\begin{equation}
\label{sl7}
    \p_t \log g(z, t)
    =
    -\frac{1}{2}\left (
      \frac{g(z, t)+e^{i\xi (t)}}{g(z, t)-e^{i\xi (t)}}
     +\frac{g(z, t)+e^{-i\xi (t)}}{g(z, t)-e^{-i\xi (t)}}
    \right )
    \p_t \log r, 
\end{equation}
which is the symmetric radial L\"owner equation (\ref{or6}) ($t\to
\lambda$, $g(z, t)\to w(z)$). Thus we have shown that a family of
conformal mappings $g(z,t)$ ($g(z,0)=z$) which conformally map the
exterior of the unit circle with two non-intersecting symmetric slits
$\Gamma_t$ (in the upper half plane) and $\bar \Gamma_t$ (in the lower
half plane) to the exterior of the unit circle gives a solution to this
equation.

\section{Conclusion}

In this paper we have introduced the dispersionless version 
of the constrained Toda hierarchy. It was formulated in two 
equivalent ways: a) by means of the Lax formalism and b) by means 
of the Hirota-like equation for the $F$-function (the dispersionless 
limit of logarithm of the tau-function). The geometric meaning of the 
dispersionless constrained Toda hierarchy was clarified. We have shown that
%solutions to the hierarchy describe conformal maps of domains symmetric 
%under reflection with respect to the real axis to the reference domain 
%(the unit disk). 
conformal maps of domains symmetric 
under reflection with respect to the real axis to the reference domain 
(the unit disk) give solutions to the hierarchy.

We have also studied finite-dimensional reductions 
of the dispersionless constrained Toda hierarchy. It was shown that the
consistency condition for one-variable reduction is a differential 
equation of the L\"owner type which we call the symmetric radial
L\"owner equation. The geometric meaning of the latter was clarified. 
We have shown that solutions to this equation are conformal maps 
of the exterior of the unit circle with two symmetric curved slits
to the exterior of the unit circle.

\appendix

\section{Appendix: Proof of the symmetric radial L\"owner equation}
\label{app:proof-rad-sym-loewner}

In this appendix we derive the symmetric radial L\"owner equation for an
evolution family, a two-parameter family of conformal mappings,
rigorously. In particular extendability of the maps is discussed
here. The last part of the proof is almost the same as the arguments in
\secref{section:symmetric} but with some omitted details.

\medskip
Assume that a Jordan curve $\gamma$ is in the upper half plane and tends
to $\infty$:
\begin{equation}
    \gamma: [0,+\infty) \to \Comp, \quad 
    \Im \gamma(t)>0, \quad \lim_{t\to\infty} \gamma(t) = \infty.
\label{gamma}
\end{equation}
We denote its complex conjugate by $\bar\gamma$:
\begin{equation}
    \bar\gamma(t) := \overline{\gamma(t)}.
\label{gamma-bar}
\end{equation}
Let $D_t$ ($t\geq 0$) be the following domain, which is symmetric with
respect to the real axis:
\begin{equation}
    D_t 
    := 
    \Comp \setminus 
    \bigl( \gamma([t,+\infty)) \cup \bar\gamma([t,+\infty) ) \bigr).
\label{Dt}
\end{equation}

As in \secref{section:symmetric}, the following lemma is essential.

\begin{lemma}
\label{lem:reflection}
 Let $D$ be a simply-connected domain in $\Comp$ symmetric with respect
 to the real axis: $z\in D \Longleftrightarrow \bar z \in D$. For
 simplicity we assume that $0\in D$. Then the holomorphic bijection
 $f:\bbD\to D$ normalised by $f(0)=0$ and $f'(0)>0$, the existence of
 which is guaranteed by Riemann's mapping theorem, satisfies
\begin{equation}
    f(\bar z) = \overline{f(z)}
\label{f:symmetric}
\end{equation}
 for $z\in D$.
\end{lemma}

\begin{proof}
 Note that $\tilde f(z) := \overline{f(\bar z)}$ is also a holomorphic
 bijection between $\bbD$ and $D$, satisfying the normalisation
 conditions $\tilde f(0)=0$, $\tilde f'(0)>0$. Since the Riemann mapping
 theorem claims that such a mapping is unique, $\tilde f = f$, which
 proves \eqref{f:symmetric}.
\end{proof}

We denote the normalised conformal mapping from the unit open disk
$\bbD$ to $D_t$ by $f_t$:
\begin{equation}
    f_t: \bbD \to D_t, \qquad
    f_t(0)=0, \ f'_t(0) > 0.
\label{ft}
\end{equation}
Because of \lemref{lem:reflection} this mapping function satisfies
\begin{equation}
    f_t(\bar z) = \overline{f_t(z)}.
\label{ft:reflection}
\end{equation}

The argument in \S3.3 \cite{Duren} or \S3.2, Chapter 3 of
\cite{kus-sug:19} can be applied to the family of domains $\{D_t\}$ and
we can reparametrise it so that $f_t$ satisfies $f_t'(0)=e^t$:
\begin{equation}
    f_t(z) = e^t z + a_2(t)\, z^2 + a_3(t)\, z^3 + \dotsb.
\label{ft:expansion}
\end{equation}
(The choice $f_t'(0)=e^t$ is just for simplicity. We can take any smooth
positive increasing function $r(t)$ as in \secref{section:symmetric}
instead of $e^t$.) We call the family $\{f_t(z)\}_{t\geq0}$ the {\em
symmetric L\"owner chain}.

For non-negative real numbers $s$ and $t$ satisfying $0\leq s<t<+\infty$
we define the {\em symmetric evolution family}
$\omega_{s,t}:\bbD\to\bbD$ by $\omega_{s,t}:= f_t^{-1}\circ f_s$, which
satisfies $\omega_{s,t}(0)=0$, $f_s=f_t\circ\omega_{s,t}$,
\begin{equation}
    \omega_{s,t}(z) = e^{s-t} z + b_2(t)\, z^2 + b_3(z)\, z^3 + \dotsb,
\label{omegast:expansion}
\end{equation}
and is univalent. Each $\omega_{s,t}$ also satisifes the conditions of
\lemref{lem:reflection}:
\begin{equation}
    \omega_{s,t}(\bar z) = \overline{\omega_{s,t}(z)}.
\label{omegast:reflection}
\end{equation}

\begin{theorem}
\label{thm:sym-rad-loewner}
%
% The function $f(z,t):=f_t(z)$ is of class $C^1$ on
% $\bbD\times[0,+\infty)$ and satisfies a partial differential equation
%\begin{equation}
%    \frac{\der f}{\der t}
%    =
%    - \frac{z}{2} \frac{\der f}{\der z}
%    \left(
%     \frac{e^{i\theta(t)}+z}{e^{i\theta(t)}-z}
%     +
%     \frac{e^{-i\theta(t)}+z}{e^{-i\theta(t)}-z}
%    \right)
%\label{sym-rad-loewner:partial}
%\end{equation}
% where $\theta(t)$ is a real continuous function on $[0,+\infty)$. 
 For a fixed $s\geq0$ and $z\in\bbD$ the function
 $\omega(t)=\omega_{s,t}(z)$ satisfies an ordinary differential equation
\begin{equation}
    \frac{d\omega}{dt}
    =
    - \frac{\omega}{2}
    \left(
     \frac{e^{i \theta(t)}+\omega}{e^{i \theta(t)}-\omega}
     +
     \frac{e^{-i\theta(t)}+\omega}{e^{-i\theta(t)}-\omega}
    \right)
\label{sym-rad-loewner}
\end{equation}
 on $t\in[s,\infty)$ with the initial value condition $w(s)=z$. 
 Here $\theta(t)$ is a real continuous function on $[0,+\infty)$. 

% For a fixed $s\geq0$ on any compact set in $\bbD$ $e^t\,
% \omega_{s,t}(z)$ uniformly converges to $f_s$ when $t\to\infty$.
%\begin{equation}
%    e^t\, \omega_{s,t}(z) \underset{t\to\infty}{\rightrightarrows}
%    f_s(z).
%\label{etomegast=>fs}
%\end{equation}
%
\end{theorem}

By changing the variables $t$, $z$ and $\omega_{s,t}$ by $\lambda:=-t$,
$\zeta:=z^{-1}$, $w(\zeta,\lambda):=1/\omega_{s,t}(1/\zeta)$, we obtain
a conformal mapping from $\Comp\setminus\bbD$ to $\Comp\setminus\bbD$,
which satisfies
\begin{equation}
    \frac{dw}{d\lambda}
    =
    -
    \frac{w}{2}
    \left(
     \frac{w+e^{-i\theta(\lambda)}}{w-e^{-i\theta(\lambda)}}
     +
     \frac{w+e^{i \theta(\lambda)}}{w-e^{i \theta(\lambda)}}
    \right),
\label{sym-rad-loewner:out}
\end{equation}
which is nothing but \eqref{or6}.

\subsection{Behaviour of functions on the boundaries}
\label{subsec:boundary-behaviour}

In this subsectin we prove extendability of $\omega_{s,t}$ to the
boundary of $\bbD$. The variable $t$ is fixed. The argument almost
follows that in \S3.2, Chapter 3 of \cite{kus-sug:19}, which is an
improved version of the argument in \cite{Duren}.

If $0\leq s<t$, $D_t\setminus D_s = \gamma([s,t)) \cup
\bar\gamma([s,t))$. Hence by \eqref{omegast:reflection} the image of
$\omega_{s,t}$ is 
\[
    \omega_{s,t}(\bbD) 
    = 
    \bbD \setminus (\alpha_{s,t} \cup \overline{\alpha_{s,t}}),
\]
where $\alpha_{s,t}:=f_t^{-1}\bigl(\gamma([s,t))\bigr)$ and
$\overline{\alpha_{s,t}}$ is its complex conjugate. 

Let us show that each arc $\alpha_{s,t}$ has an endpoint $\lambda(t)$ on
$\der\bbD$, $\displaystyle\lambda(t)=\lim_{s\nearrow t}
f_t^{-1}(\gamma(s))$, which does not depend on $s\in[0,t)$.

We apply the Carath\'eodory extension theorem by opening the slit
$$\gamma([t,\infty)) \cup \bar\gamma([t,\infty]) \cup \{\infty\}$$ in the
Riemann sphere $\widehat\Comp$ by square root\footnote{In the derivation
of L\"owner's equation in \S3.3 of \cite{Duren} Carath\'eodory's theorem
is applied to a domain whose boundary is not a Jordan curve. In
\cite{kus-sug:19} this gap is filled by opening the slit by square
root. We use this idea here.}. We take a branch of
\begin{equation}
    \zeta = F(w)
    :=
    \sqrt{
     \frac{1-w/\gamma(t)}{1-w/\bar\gamma(t)}
    }
\label{F(w)}
\end{equation}
satisfying $F(0)=1$. Its inverse map is
\begin{equation}
    w = G(\zeta)
    :=
    \gamma(t)\,\bar\gamma(t)
    \frac{1-\zeta^2}{\bar\gamma(t) - \gamma(t)\,\zeta^2}.
\label{G(zeta)}
\end{equation}
Endpoints of the slit are mapped by $F$ to $F(\gamma(t))=0$ and
$F(\bar\gamma(t))=\infty$. The image of the slit is a Jordan closed
curve $\Gamma^+ \cup \{0\} \cup \Gamma^- \cup \{\infty\}$ in
$\widehat\Comp$, where $\Gamma^\pm$ are two connected components of
$F\bigl(\gamma((t,+\infty))\bigr)$ (\figref{fig:extension-ft}).

\begin{figure}[h]
 \centering
 \def\svgwidth{.8\linewidth}
 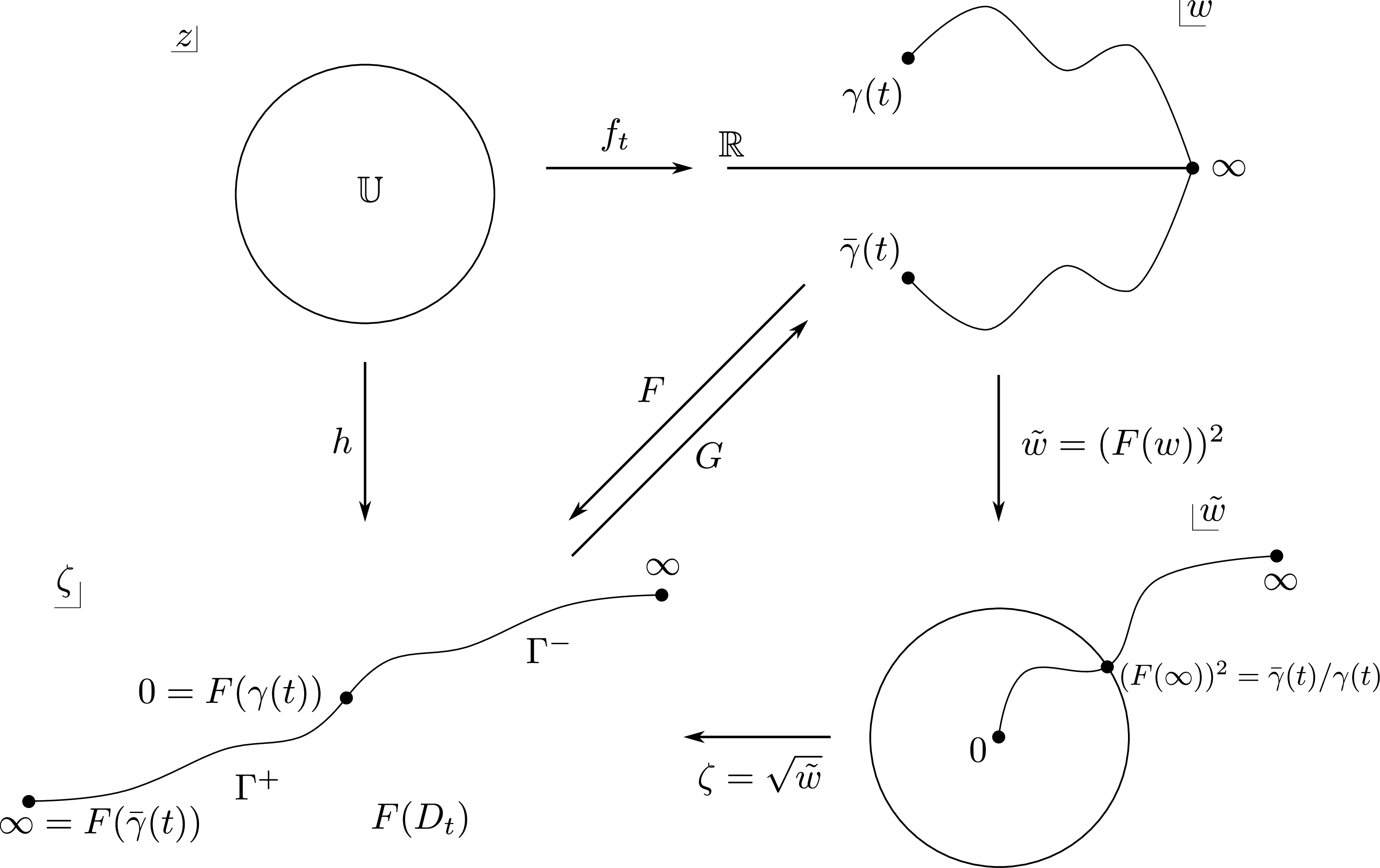
 \caption{Conformal map $F\circ f_t:\bbD\to F(D_t)$ is extendable to
 $\overline\bbD\to\overline{F(D_t)}$.}
\label{fig:extension-ft}
\end{figure}

Applying the Carath\'eodory extension theorem to $F\circ f_t$, we obtain
a homeomorphism $h:\overline\bbD\to\overline{F(D_t)}$ which is
holomorphic in $\bbD$. Thus we obtain an extension of $f_t$ on
$\overline\bbD$: $f_t = G\circ h: \overline\bbD \to
\overline{D_t}=\widehat\Comp$. (We use the same notation for the
extended map as the original one.)

Since $f_t=F^{-1}\circ h$ on $\bbD$, we have
\[
    f_t^{-1}(\gamma(s)) = h^{-1}(F(\gamma(s))), \quad
    f_t^{-1}(\bar\gamma(s)) = h^{-1}(F(\bar\gamma(s)))
\]
for $0\leq s<t$. Hence, when $s$ approaches to $t$ from below, there
exists limits,
\begin{align}
    \lambda(t)
    &:=
    \lim_{s\nearrow t} f_t^{-1}(\gamma(s)) 
    =
    \lim_{s\nearrow t} h^{-1}(F(\gamma(s)))
    =
    h^{-1}(F(\gamma(t))) = h^{-1}(0),
\label{def:lambda}
\\
    \bar\lambda(t)
    &:=
    \lim_{s\nearrow t} f_t^{-1}(\bar\gamma(s)) 
    =
    \lim_{s\nearrow t} h^{-1}(F(\bar\gamma(s)))
    =
    h^{-1}(F(\bar\gamma(t))) = h^{-1}(\infty),
\label{def:lambdabar}
\end{align}
and, due to \eqref{gamma-bar} and \eqref{ft:reflection}, they are
complex conjugate to each other,
\begin{equation}
    \bar\lambda(t) = \overline{\lambda(t)}.
\label{bar(lambda)=lambdabar}
\end{equation}
The proof of (right/left) continuity of $\lambda(t)$ is the same as that
in \S3.3 of \cite{Duren}.

As in \figref{fig:Jtu} we denote the counterclockwise arc from
$\lambda(t)$ to $\bar\lambda(t)$ by $I^+=I^+_t$ and the other arc by
$I^-=I^-_t$: $h(I^\pm)=\Gamma^{\pm}$. Restrictions of
$f_t:\bar\bbD\to\widehat\Comp$ on $I^\pm$, 
$
    f_t^\pm: 
    I^\pm \to
    \gamma([t,+\infty))\cup\{\infty\} \cup \bar\gamma([t,+\infty))
$,
are homeomorphisms. Therefore for any $u \in (t,+\infty)$ the subset
$J_{t,u}:=J_{t,u}^+\cup J_{t,u}^-$ of $\der\bbD$, where
$J^\pm_{t,u}:=(f_t^\pm)^{-1}\bigl(\gamma([t,u))\bigr)$, is an open
arc of $\der\bbD$. 

\begin{figure}[h]
 \centering
 \def\svgwidth{.8\linewidth}
 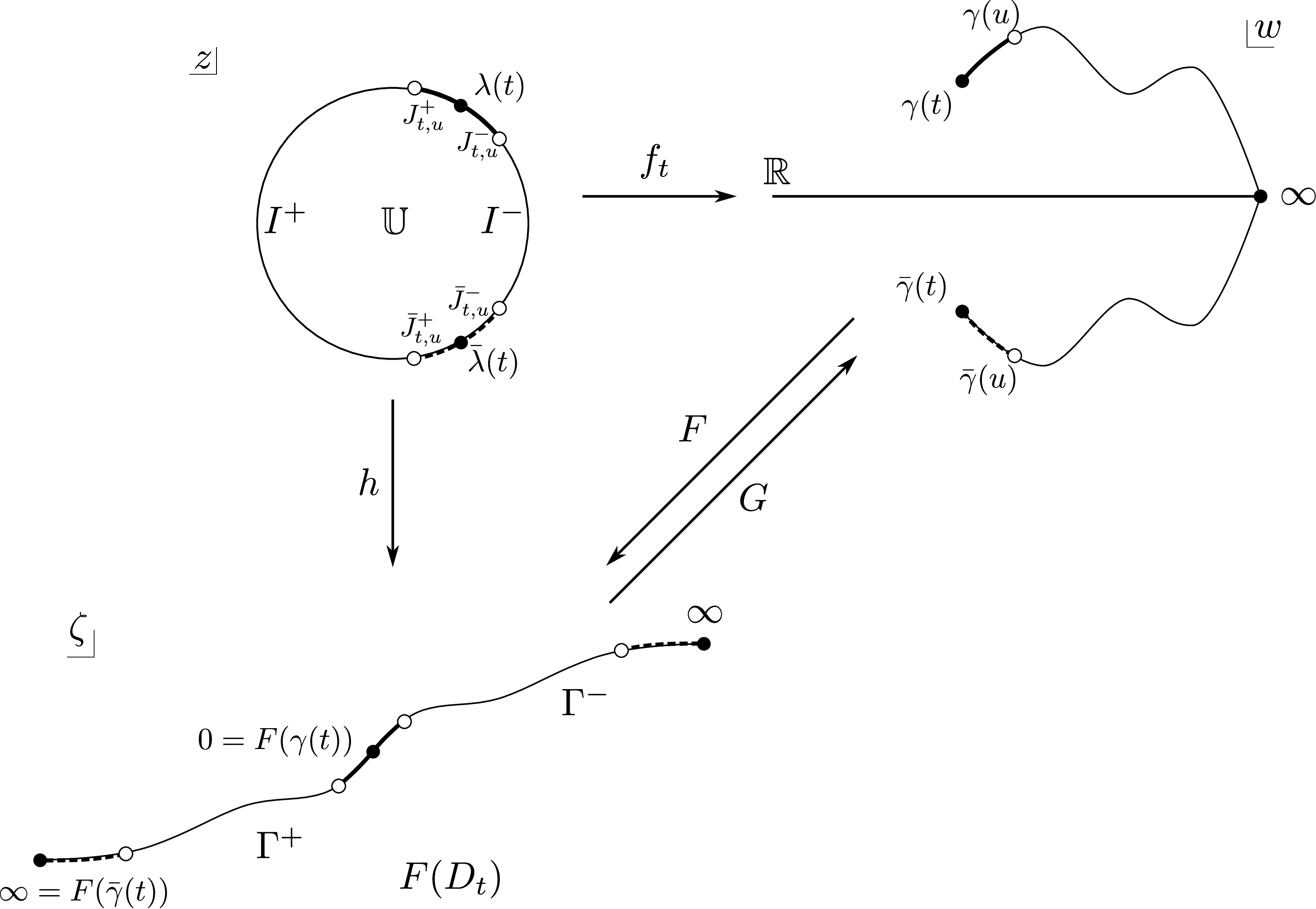
 \caption{$f_t$ gives maps from arcs $J^\pm_{t,u}$ and 
 $\bar J^\pm_{t,u}$ to $\gamma([t,u))$ and $\bar\gamma([t,u))$.}
\label{fig:Jtu}
\end{figure}

Similarly, the subset $\bar J_{t,u}:=\bar J_{t,u}^+\cup \bar J_{t,u}^-$
of $\der\bbD$, where $\bar J^\pm_{t,u}:=(f_t^\pm)^{-1}\bigl(\bar
\gamma([t,u))\bigr)$, is an open arc of $\der\bbD$ and, because of
\eqref{ft:reflection}, 
$
    \bar J_{t,u}^{\pm} = \overline{J_{t,u}^\pm}
$,
$
    \bar J_{t,u} = \overline{J_{t,u}}
$.

Then $f_t(J_{t,u}) = \gamma([t,u))$ and $f_t(\bar
J_{t,u})=\bar\gamma([t,u))$ and, if $0\leq t < u$,
\[
    I_t^\pm \setminus (J_{t,u}^\pm \cup \bar J_{t,u}^\pm)
    =
    f_t^{-1}\bigl(
     \gamma([u,\infty)) \cup \bar\gamma([u,\infty))
    \bigr)
    \cap
    I_t^\pm.
\]
Therefore, if we define $\omega_{t,u}$ on the boundary $\der\bbD$ as
\begin{equation}
    \omega_{t,u}(z) :=
    \begin{cases}
     f_u^{-1} \bigl(f_t(z)\bigr), 
     &(z\in \bbD \cup J_{t,u} \cup \bar J_{t,u}),
    \\
     (f_u^\pm)^{-1} \bigl(f_t(z)\bigr), 
     &(z\in I_t^\pm \setminus (J_{t,u} \cup \bar J_{t,u})),
    \end{cases}
\label{omegast:boundary}
\end{equation}
it is a continuous map $\omega_{t,u}:\bar\bbD\to\bar\bbD$. This is the
desired extension of the symmetric evolution family $\omega_{t,u}(z)$.

The images of subsets $\bbD$, $J_{t,u}$ and $\bar J_{t,u}$ of $\bar\bbD$
are
\[
    \omega_{t,u}(\bbD)=\bbD\setminus(\alpha_{t,u}\cup\bar\alpha_{t,u}),
    \quad 
    \omega_{t,u}(J_{t,u})= \alpha_{t,u},
    \quad 
    \omega_{t,u}(\bar J_{t,u})= \bar\alpha_{t,u},
\]
respectively, as shown in \figref{fig:omegatu(alphatu)}.

\begin{figure}[h]
 \centering
 \def\svgwidth{.8\linewidth}
 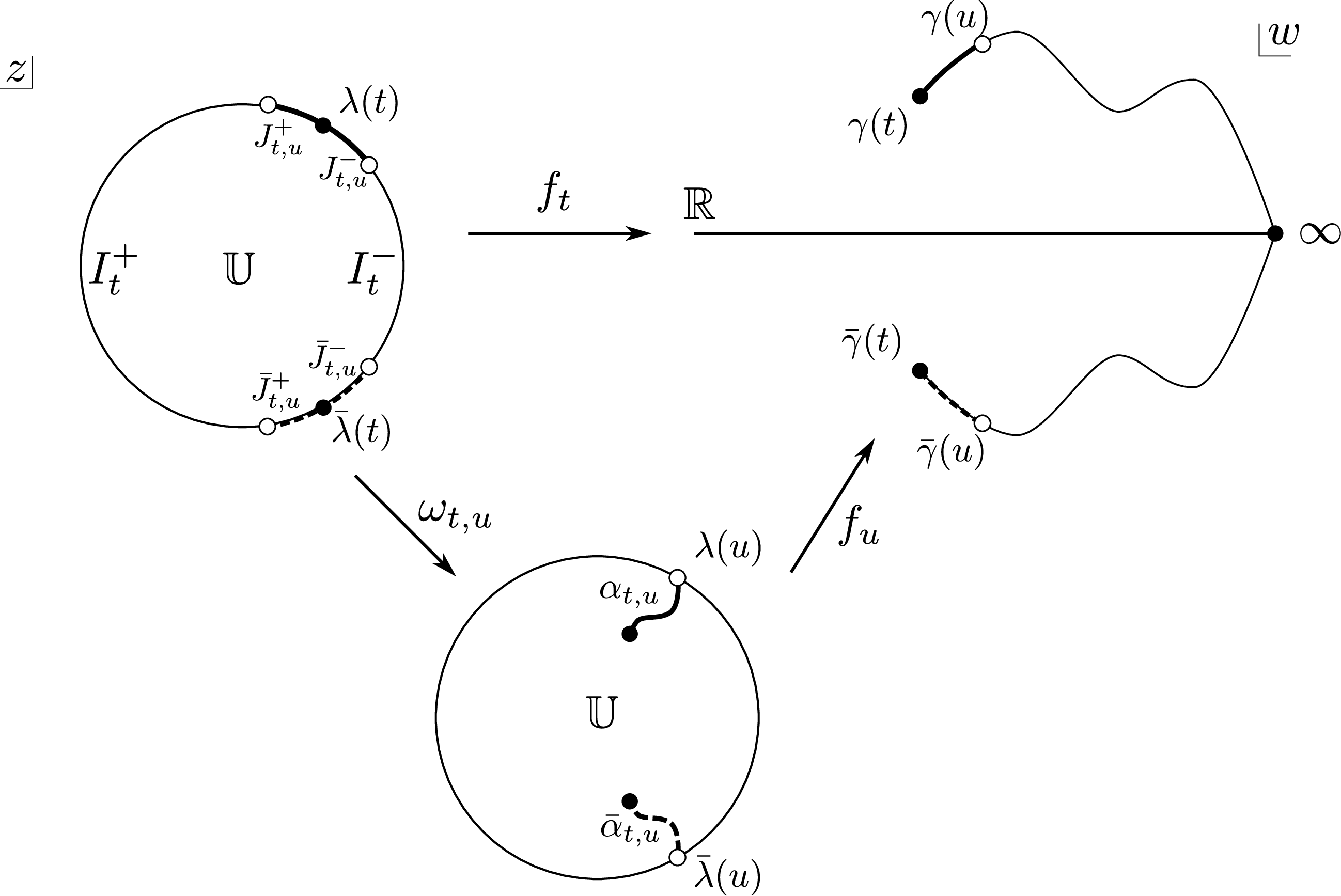
 \caption{$\omega_{t,u}$ maps arcs $J^\pm_{t,u}$ and
 $\bar J_{t,u}^\pm$ to arcs $\alpha_{t,u}$ and $\bar \alpha_{t,u}$.}
\label{fig:omegatu(alphatu)}
\end{figure}

\subsection{Symmetric radial L\"owner equation for $\{\omega_{s,t}(z)\}$.}
\label{subsec:sym-rad-loewner}

We add some details in derivation of the symmetric radial L\"owner
equation discussed in \secref{section:symmetric}. (Recall that in this
appendix $r(t)$ is normalized as $r(t)=e^t$ for simplicity.) We follow
the argument in \S3.2, Chapter 3 of \cite{kus-sug:19}.

As $\omega_{t,u}(z)=e^{t-u}z + O(z^2)$, we can take a branch of 
\begin{equation}
    h(z) = h(z;t,u):=
    \log\frac{\omega_{t,u}(z)}{z},
\label{def:h}
\end{equation}
such that $h(0)=t-u$. This function is holomorphic in $\bbD$ and
continuous on its closure $\bar\bbD$. Since
$
    \omega_{t,u}(\der\bbD\setminus(J_{t,u}\cup\bar J_{t,u}))
    \subset
    \der\bbD
$,
$\Re h(z) = 0$, if $z\in\der\bbD\setminus(J_{t,u}\cup\bar J_{t,u})$, and
since $\omega_{t,u}(J_{t,u}\cup\bar J_{t,u})\subset\bbD$, $\Re h(z)<0$,
if $z\in J_{t,u}\cup\bar J_{t,u}$.  Hence, applying the Schwarz integral
formula for the unit disk\footnote{Recall that formula \eqref{sl2} is
for functions outside of the unit disk.},
\begin{equation*}
    f(z) = 
    \frac{1}{2\pi}
    \int_{-\pi}^{\pi} 
     \Re f(e^{i\theta}) 
     \frac{e^{i\theta}+z}{e^{i\theta}-z}
    d\theta + i \Im f(0),\qquad (z\in\bbD),
\end{equation*}
to $h(z)$, we have
\begin{equation}
    h(z) 
    =
    \frac{1}{2\pi} \left(
     \int_{a}^{b} + \int_{-b}^{-a}
      \Re h(e^{i\theta}) 
      \frac{e^{i\theta}+z}{e^{i\theta}-z}
     d\theta
    \right),
\label{Phi=integral}
\end{equation}
where $a$ and $b$ are real numbers such that $e^{ia}$ and $e^{ib}$ are
endpoints of $J_{t,u}$ and, correspondingly, $e^{-ib}$ and $e^{-ia}$ are
endpoints of $\bar J_{t,u}$. Note that $h(0) = t-u \in \Real$ and
therefore $\Im h(0)=0$.

As in \secref{section:symmetric}, because of the symmetry $h(\bar
z)=\overline{h(z)}$ the above formula is rewritten as 
\begin{equation}
    h(z) 
    =
    \frac{1}{\pi}
    \int_{a}^{b}
     \Re h(e^{i\theta}) K(e^{i\theta},z)
    d\theta,
\label{Phi=integral:temp}
\end{equation}
where $K(\zeta,z)$ is the function defined by
$\displaystyle
    K(z,\zeta) 
    :=
    \frac{1}{2}
    \left(
     \frac{\zeta+z}{\zeta-z} + \frac{\bar\zeta+z}{\bar\zeta-z}
    \right)
$.
Substituting $\omega_{s,t}(z)$ for $z$ in \eqref{Phi=integral:temp} and
using $\omega_{t,u}(\omega_{s,t}(z)) = \omega_{s,u}(z)$, we have
\begin{equation}
    \log\frac{\omega_{s,u}(z)}{\omega_{s,t}(z)}
    =
    \frac{1}{\pi}
    \int_{a}^{b}
     \Re h(e^{i\theta}) K_R(e^{i\theta}, \omega_{s,t}(z))
    d\theta
    +
    \frac{i}{\pi}
    \int_{a}^{b}
     \Re h(e^{i\theta}) K_I(e^{i\theta}, \omega_{s,t}(z))
    d\theta,
\label{log(omega/omega)=integral:temp}
\end{equation}
where $K_R(\zeta,z):=\Re K(\zeta,z)$ and $K_I(\zeta,z):=\Im
K(\zeta,z)$. From the mean value theorem of integrals it follows that
there exist $\xi$ and $\eta$ in the interval $(a,b)$ such that
\begin{align*}
    \int_{a}^{b}
     \Re h(e^{i\theta}) K_R(e^{i\theta}, \omega_{s,t}(z))
    d\theta
    &=
    K_R(e^{i\xi},\omega_{s,t}(z))
    \int_{a}^{b} \Re h(e^{i\theta})d\theta,
\\
    \int_{a}^{b}
     \Re h(e^{i\theta}) K_I(e^{i\theta}, \omega_{s,t}(z))
    d\theta
    &=
    K_I(e^{i\eta},\omega_{s,t}(z))
    \int_{a}^{b} \Re h(e^{i\theta})d\theta.
\end{align*}
Hence by dividing \eqref{log(omega/omega)=integral:temp} by
$\displaystyle
   t-u = h(0)
   =
   \frac{1}{\pi}
    \int_{a}^{b}
     \Re h(e^{i\theta}) \,
    d\theta,
$
we obtain
\[
    \frac{\log \omega_{s,u}(z)- \log \omega_{s,t}(z)}{u-t}
    =
    -\bigl(K_R(e^{i\xi},\omega_{s,t}(z)) 
     + i\, K_I(e^{i\eta},\omega_{s,t}(z))\bigr). 
\]
When $u\searrow t$, $J_{t,u}$ contracts to $\lambda(t)$ (cf.\
\figref{fig:Jtu}) and therefore both of $e^{i\xi}$ and $e^{i\eta}$ tend
to $\lambda(t)$. Thus
\begin{equation}
    \frac{\der}{\der t} \log \omega_{s,t}(z)
    =
    - K(\lambda(t),\omega_{s,t}(z)),
\label{sym-rad-loewner:right/left}
\end{equation}
where the left hand side is regarded as the right derivative. Similarly
the limit $t\nearrow u$ gives \eqref{sym-rad-loewner:right/left} as the
left derivative, because $J_{t,u}$ contracts to $\lambda(u)$ in this
limit. Thus we have proved the symmetric radial L\"owner equation
\eqref{sym-rad-loewner} for $\omega(t)=\omega_{s,t}(z)$,
$e^{i\theta(t)}=\lambda(t)$.

\section*{Acknowledgements}

\addcontentsline{toc}{section}{\hspace{6mm}Acknowledgements}

We are grateful to Ikkei Hotta for his detailed comments on complex
analysis.
The results of T.T. (Sections 6,7 and Appendix) 
is an output of a research project implemented as part of
the Basic Research Program at the National Research University Higher
School of Economics.
The work of A.Z. (Sections 2--5) 
was performed at the Steklov Mathematical Institute of 
Russian Academy of Sciences, Moscow, and was supported by the 
Russian Science Foundation under grant 19-11-00062.

\end{document}

%% file: g_zt_.eps_tex
%% Creator: Inkscape inkscape 0.92.5, www.inkscape.org
%% Accompanies image file 'g_zt_.eps' (pdf, eps, ps)
\begingroup%
  \makeatletter%
  \providecommand\color[2][]{%
    \errmessage{(Inkscape) Color is used for the text in Inkscape, but the package 'color.sty' is not loaded}%
    \renewcommand\color[2][]{}%
  }%
  \providecommand\transparent[1]{%
    \errmessage{(Inkscape) Transparency is used (non-zero) for the text in Inkscape, but the package 'transparent.sty' is not loaded}%
    \renewcommand\transparent[1]{}%
  }%
  \providecommand\rotatebox[2]{#2}%
  \newcommand*\fsize{\dimexpr\f@size pt\relax}%
  \newcommand*\lineheight[1]{\fontsize{\fsize}{#1\fsize}\selectfont}%
  \ifx\svgwidth\undefined%
    \setlength{\unitlength}{1282.49997165bp}%
    \ifx\svgscale\undefined%
      \relax%
    \else%
      \setlength{\unitlength}{\unitlength * \real{\svgscale}}%
    \fi%
  \else%
    \setlength{\unitlength}{\svgwidth}%
  \fi%
  \global\let\svgwidth\undefined%
  \global\let\svgscale\undefined%
  \makeatother%
  \begin{picture}(1,0.38596494)%
    \lineheight{1}%
    \setlength\tabcolsep{0pt}%
    \put(0,0){\includegraphics[width=\unitlength]{g_zt_.eps}}%
  \end{picture}%
\endgroup%

%% file: h_zst_.eps_tex
%% Creator: Inkscape inkscape 0.92.5, www.inkscape.org
%% Accompanies image file 'h_zst_.eps' (pdf, eps, ps)
\begingroup%
  \makeatletter%
  \providecommand\color[2][]{%
    \errmessage{(Inkscape) Color is used for the text in Inkscape, but the package 'color.sty' is not loaded}%
    \renewcommand\color[2][]{}%
  }%
  \providecommand\transparent[1]{%
    \errmessage{(Inkscape) Transparency is used (non-zero) for the text in Inkscape, but the package 'transparent.sty' is not loaded}%
    \renewcommand\transparent[1]{}%
  }%
  \providecommand\rotatebox[2]{#2}%
  \newcommand*\fsize{\dimexpr\f@size pt\relax}%
  \newcommand*\lineheight[1]{\fontsize{\fsize}{#1\fsize}\selectfont}%
  \ifx\svgwidth\undefined%
    \setlength{\unitlength}{1289.9999207bp}%
    \ifx\svgscale\undefined%
      \relax%
    \else%
      \setlength{\unitlength}{\unitlength * \real{\svgscale}}%
    \fi%
  \else%
    \setlength{\unitlength}{\svgwidth}%
  \fi%
  \global\let\svgwidth\undefined%
  \global\let\svgscale\undefined%
  \makeatother%
  \begin{picture}(1,0.92441867)%
    \lineheight{1}%
    \setlength\tabcolsep{0pt}%
    \put(0,0){\includegraphics[width=\unitlength]{h_zst_.eps}}%
  \end{picture}%
\endgroup%

%% file: extension-ft.eps_tex
%% Creator: Inkscape inkscape 0.92.5, www.inkscape.org
%% Accompanies image file 'extension-ft.eps' (pdf, eps, ps)
\begingroup%
  \makeatletter%
  \providecommand\color[2][]{%
    \errmessage{(Inkscape) Color is used for the text in Inkscape, but the package 'color.sty' is not loaded}%
    \renewcommand\color[2][]{}%
  }%
  \providecommand\transparent[1]{%
    \errmessage{(Inkscape) Transparency is used (non-zero) for the text in Inkscape, but the package 'transparent.sty' is not loaded}%
    \renewcommand\transparent[1]{}%
  }%
  \providecommand\rotatebox[2]{#2}%
  \newcommand*\fsize{\dimexpr\f@size pt\relax}%
  \newcommand*\lineheight[1]{\fontsize{\fsize}{#1\fsize}\selectfont}%
  \ifx\svgwidth\undefined%
    \setlength{\unitlength}{1600.43340549bp}%
    \ifx\svgscale\undefined%
      \relax%
    \else%
      \setlength{\unitlength}{\unitlength * \real{\svgscale}}%
    \fi%
  \else%
    \setlength{\unitlength}{\svgwidth}%
  \fi%
  \global\let\svgwidth\undefined%
  \global\let\svgscale\undefined%
  \makeatother%
  \begin{picture}(1,0.6291076)%
    \lineheight{1}%
    \setlength\tabcolsep{0pt}%
    \put(0,0){\includegraphics[width=\unitlength]{extension-ft.eps}}%
  \end{picture}%
\endgroup%

%% file: Jtu.eps_tex
%% Creator: Inkscape inkscape 0.92.5, www.inkscape.org
%% Accompanies image file 'Jtu.eps' (pdf, eps, ps)
\begingroup%
  \makeatletter%
  \providecommand\color[2][]{%
    \errmessage{(Inkscape) Color is used for the text in Inkscape, but the package 'color.sty' is not loaded}%
    \renewcommand\color[2][]{}%
  }%
  \providecommand\transparent[1]{%
    \errmessage{(Inkscape) Transparency is used (non-zero) for the text in Inkscape, but the package 'transparent.sty' is not loaded}%
    \renewcommand\transparent[1]{}%
  }%
  \providecommand\rotatebox[2]{#2}%
  \newcommand*\fsize{\dimexpr\f@size pt\relax}%
  \newcommand*\lineheight[1]{\fontsize{\fsize}{#1\fsize}\selectfont}%
  \ifx\svgwidth\undefined%
    \setlength{\unitlength}{1455.43881336bp}%
    \ifx\svgscale\undefined%
      \relax%
    \else%
      \setlength{\unitlength}{\unitlength * \real{\svgscale}}%
    \fi%
  \else%
    \setlength{\unitlength}{\svgwidth}%
  \fi%
  \global\let\svgwidth\undefined%
  \global\let\svgscale\undefined%
  \makeatother%
  \begin{picture}(1,0.69053872)%
    \lineheight{1}%
    \setlength\tabcolsep{0pt}%
    \put(0,0){\includegraphics[width=\unitlength]{Jtu.eps}}%
  \end{picture}%
\endgroup%

%% file: alphatu.eps_tex
%% Creator: Inkscape inkscape 0.92.5, www.inkscape.org
%% Accompanies image file 'alphatu.eps' (pdf, eps, ps)
\begingroup%
  \makeatletter%
  \providecommand\color[2][]{%
    \errmessage{(Inkscape) Color is used for the text in Inkscape, but the package 'color.sty' is not loaded}%
    \renewcommand\color[2][]{}%
  }%
  \providecommand\transparent[1]{%
    \errmessage{(Inkscape) Transparency is used (non-zero) for the text in Inkscape, but the package 'transparent.sty' is not loaded}%
    \renewcommand\transparent[1]{}%
  }%
  \providecommand\rotatebox[2]{#2}%
  \newcommand*\fsize{\dimexpr\f@size pt\relax}%
  \newcommand*\lineheight[1]{\fontsize{\fsize}{#1\fsize}\selectfont}%
  \ifx\svgwidth\undefined%
    \setlength{\unitlength}{1245.92692011bp}%
    \ifx\svgscale\undefined%
      \relax%
    \else%
      \setlength{\unitlength}{\unitlength * \real{\svgscale}}%
    \fi%
  \else%
    \setlength{\unitlength}{\svgwidth}%
  \fi%
  \global\let\svgwidth\undefined%
  \global\let\svgscale\undefined%
  \makeatother%
  \begin{picture}(1,0.66825208)%
    \lineheight{1}%
    \setlength\tabcolsep{0pt}%
    \put(0,0){\includegraphics[width=\unitlength]{alphatu.eps}}%
  \end{picture}%
\endgroup%